\documentclass[reqno]{amsart}
\usepackage{amssymb}
\usepackage[pdftex]{graphicx}
\usepackage
[unicode=true,pdfusetitle,
 bookmarks=true,bookmarksnumbered=false,bookmarksopen=false,
 breaklinks=false,pdfborder={0 0 1},backref=false,colorlinks=false]
 {hyperref}

\makeatletter
\numberwithin{equation}{section}
\theoremstyle{plain}
\newtheorem{thm}{\protect\theoremname}
  \theoremstyle{plain}
  \newtheorem{prop}[thm]{\protect\propositionname}
  \theoremstyle{remark}
  \newtheorem{rem}[thm]{\protect\remarkname}
  \theoremstyle{plain}
  \newtheorem{lem}[thm]{\protect\lemmaname}

\newcommand{\arc}[1]{\mathrm{arc}(#1)}
\makeatother

  \providecommand{\lemmaname}{Lemma}
  \providecommand{\propositionname}{Proposition}
  \providecommand{\remarkname}{Remark}
\providecommand{\theoremname}{Theorem}

\begin{document}

\title[Long-Time Asymptotics for IDNLS]{Long-time asymptotics for the integrable discrete nonlinear Schr\"odinger
equation: the focusing case}

\author{Hideshi Yamane }

\date{\today}

\address{Department of Mathematical Sciences, Kwansei Gakuin University, Gakuen
2-1 Sanda, Hyogo 669-1337, Japan}

\email{yamane@kwansei.ac.jp}

\thanks{This work was partially supported by JSPS KAKENHI Grant Number 26400127. }



\maketitle
%\subjclass[2000]{Primary }
%    For articles to be published after 1 January 2010, you may use
%    the following version:

%\dedicatory{}

\begin{abstract}
We investigate the long-time asymptotics for the focusing integrable
discrete nonlinear Schr\"odinger  equation. Under generic assumptions
on the initial value, the solution is asymptotically a sum of 1-solitons.
We find different phase shift formulas in different regions. Along
rays away from solitons, the behavior of the solution is decaying
oscillation. This is one way of stating the soliton resolution conjecture.
The proof is based on the nonlinear steepest descent method. 
\end{abstract}

\section{Introduction}

In this article we study the long-time behavior of the solutions to
the focusing integrable discrete nonlinear Schr\"odinger  equation
(IDNLS) introduced by Ablowitz and Ladik (\cite{AblowitzLadikJMP76})
on the doubly infinite lattice (i.e. $n\in\mathbb{Z}$): 
\begin{equation}
i\frac{d}{dt}R_{n}+(R_{n+1}-2R_{n}+R_{n-1})+|R_{n}|^{2}(R_{n+1}+R_{n-1})=0.\label{eq:IDNLS}
\end{equation}
It is a discrete version of the focusing nonlinear Schr\"odinger
 equation (NLS) 
\[
iu_{t}+u_{xx}+2u|u|^{2}=0.
\]
The equation \eqref{eq:IDNLS} can be solved by the inverse scattering
transform (IST). Here we employ the Riemann-Hilbert formalism of IST
following \cite{APT}. Eigenvalues appear in quartets of the form
$(\pm z_{j},\pm\bar{z}_{j}^{-1})$. 

In the reflectionless case, it is well known (\cite{AblowitzLadikJMP76})
that \eqref{eq:IDNLS} admits a multi-soliton solution under generic
assumptions. When there is only one quartet of eigenvalues including
$z_{1}=\exp(\alpha_{1}+i\beta_{1})$ with $\alpha_{1}>0$, $R_{n}(t)$
is a bright 1-soliton solution, namely, 
\[
R_{n}(t)=\mathrm{BS}\left(n,t;z_{1},C_{1}(0)\right),
\]
where $C_{1}(0)$ is the norming constant and 
\begin{eqnarray*}
\mathrm{BS}(n,t;z_{1},C_{1}(0)) & = & \frac{C_{1}(0)}{|C_{1}(0)|}\exp\bigl(-i[2\beta_{1}(n+1)-2w_{1}t]\bigr)\\
 &  & \times\sinh(2\alpha_{1})\mathrm{sech}[2\alpha_{1}(n+1)-2v_{1}t-\theta_{1}].
\end{eqnarray*}
Here $\mathrm{BS}$ stands for 'bright soliton' and 
\begin{align*}
 & v_{1}=-\sinh(2\alpha_{1})\sin(2\beta_{1}),\quad w_{1}=\cosh(2\alpha_{1})\cos(2\beta_{1})-1,\\
 & \theta_{1}=\log|C_{1}(0)|-\log\sinh(2\alpha_{1}).
\end{align*}
The solution $\mathrm{BS}(n,t;z_{1},C_{1}(0))$ involves a traveling
wave with $\mathrm{sech}$ profile. We denote its velocity by $\mathrm{tw}(z_{1})$.
In other words, 
\[
\mathrm{tw}(z_{1})=\mathrm{tw}(\exp(\alpha_{1}+i\beta_{1}))=\alpha_{1}^{-1}v_{1}=-\alpha_{1}^{-1}\sinh(2\alpha_{1})\sin(2\beta_{1}).
\]

In the present paper, we study what happens if the reflection coefficient
corresponding to $R_{n}(0)$ does not vanish identically. If the quartets
of eigenvalues are $(\pm z_{j},\pm\bar{z}_{j}^{-1})$ with $\mathrm{tw}(z_{j})<\mathrm{tw}(z_{j'})\,(j<j')$,
then we have, formally, 
\begin{align*}
R_{n}(t)\sim & \sum_{j\in G_{1}}\mathrm{BS}\left(n,t;z_{j},\delta_{n/t}(0)\delta_{n/t}(z_{j})^{2}p_{j}T(z_{j})^{-2}C_{j}(0)\right)\\
 & +\sum_{j\in G_{2}}\mathrm{BS}\left(n,t;z_{j},p_{j}T(z_{j})^{-2}C_{j}(0)\right),\\
p_{j}= & \prod_{k>j}z_{k}^{2}\bar{z}_{k}^{-2},\\
T(z_{j})= & \prod_{k>j}\frac{z_{k}^{2}(z_{j}^{2}-\bar{z}_{k}^{-2})}{z_{j}^{2}-z_{k}^{2}},
\end{align*}
under generic assumptions. Here $|\mathrm{tw}(z_{j})|<2$ for $j\in G_{1}$
and $|\mathrm{tw}(z_{j})|\ge2$ for $j\in G_{2}$. See \eqref{eq:deltaintegral}
and Remark \ref{rem:delta} for the definition of $\delta_{n/t}(z)=\delta(z)$.
In the reflectionless case we have $\delta_{n/t}(0)=\delta_{n/t}(z_{j})=1$
and recover the known formula about the asymptotic behavior of a multi-soliton.
See Theorems \ref{thm:main1}, \ref{thm:main2} and \ref{thm:main3}
for details.

We review some previous results about the long-time asymptotics of
some integrable equations in the perturbed (i.e. not reflectionless)
case. In \cite{Nov}, the asymptotics for the focusing IDNLS was studied
in the solitonless case. The pioneering work \cite{DZ} established
the method of nonlinear steepest descent, which is employed in the
present paper and all the works quoted below. The defocusing NLS was
dealt with in \cite{Important}. The appearance of soliton terms in
the focusing case was observed in \cite{DP}, \cite{Fokas-Its96},
\cite{KamvissisNLS95} and\cite{KamvissisNLS96} among others. The
present author investigated the defocusing IDNLS in \cite{IDNLS,IDNLS2}.
The Toda lattice was studied in \cite{Kamvissis} under the assumption
of the absence of solitons and later in \cite{KTsolitonregion}. Our
treatment of solitons is based on the method of \cite{TodaRarefaction},
which was used in \cite{Teschl_KdV} and \cite{KTsolitonregion}.
Another way to study this kind of problems is the use of Gelfand-Levitan-Marchenko
equations (e.g. \cite{Tanaka}). 

The above mentioned works and the present article are related to a
broad statement called \textit{the soliton resolution conjecture.
}Roughly speaking, it asserts that a solution to any reasonable solution
to a (not necessarily integrable) nonlinear dispersive equation, typically
an NLS, resolves into a sum of solitons (or soliton-like states) and
a decaying  radiation part. See \cite{Tao} for a brief survey. 

The arguments in Sections 2 and 3 apply to  the half-space $t>0,\,n\in\mathbb{Z}$.
In Sections 4-7 we study the region $|n|<2t$. This is the case where
there are four distinct saddle points on $|z|=1$. In Section 8 we
treat two other regions, in which stationary points have different
configurations. 

The defocusing IDNLS admits dark solitons which satisfy non-zero boundary
conditions (\cite{AblowitzBiondiniPrinari}) in the reflectionless
case. It would be an interesting and difficult task to study its solutions
in a more general setting. 

\section{Inverse scattering transform }

In this section we explain some facts about the inverse scattering
transform for the focusing IDNLS following \cite{AblowitzLadikJMP76}
and \cite[Chap. 3]{APT}. 

First we discuss the unique solvability of the Cauchy problem for
\eqref{eq:IDNLS}.
\begin{prop}
\label{prop:ODE}Let $p$ be a non-negative integer. Assume that the
initial value $R(0)=\{R_{n}(0)\}_{n\in\mathbb{Z}}$ satisfies 
\begin{align}
 & \|R(0)\|_{1,p}=\sum_{n=-\infty}^{\infty}(1+|n|)^{p}|R_{n}(0)|<\infty.\label{eq:l1}
\end{align}
Then \eqref{eq:IDNLS} has a unique solution in $\ell^{1,p}=\{\{r_{n}\}_{n=-\infty}^{\infty}\colon\sum(1+|n|)^{p}|r_{n}|<\infty\}$
for $0\le t<\infty$.
\end{prop}
\begin{proof}
We can regard \eqref{eq:IDNLS} as an ODE in the Banach space $\ell^{1,p}\subset\ell^{\infty}$.
First we solve it in $\ell^{\infty}$. Set $c_{-\infty}=\prod_{n=-\infty}^{\infty}(1+|R_{n}|^{2})\ge1$,
$\rho=(c_{-\infty}-1)^{1/2}$. Since $1+|R_{n}(0)|^{2}\le c_{-\infty}$
for each $n$, we have $\|R(0)\|_{\infty}\le\rho$. Set \ensuremath{B:=\{R=\{R_{n}\}\in\ell^{\infty};\|R-R(0)\|_{\infty}\le\rho\}.}
 Since the right-hand side is Lipschitz continuous and bounded if
$R=\{R_{n}\}\in B$, \eqref{eq:IDNLS} can be solved in $B$ locally
in time, say up to $t=t_{1}=t_{1}(\rho)$. By a standard argument
about ODEs in a Banach space, $t_{1}$ is determined by $\rho$ only.
Since $c_{-\infty}$ and $\rho$ are conserved quantities, we have
\ensuremath{\|R(t)\|_{\infty}\le\rho}
 for $0\le t<t_{1}$. Then we solve \eqref{eq:IDNLS} again with the
initial value at $t=t_{1}/2$. The solution can be extended up to
$t=3t_{1}/2$. We repeat this process  to extend the solution $\{R_{n}(t)\}\in\ell^{\infty}$
indefinitely and it satisfies \ensuremath{\|R(t)\|_{\infty}\le\rho}
 for $0\le t<\infty$. We obtain $\|\frac{d}{dt}R(t)\|_{1,p}\le\textrm{const.}\|R(t)\|_{1,p}$.
By integration, we get \ensuremath{\|R(t)\|_{1,p}\le\|R(0)\|_{1,p}+\textrm{const.}\int_{0}^{t}\|R(\tau)\|_{1,p}\,d\tau.}
 By virtue of the Gronwall inequality, $\|R(t)\|_{1,p}$ never blows
up in a finite time. 
\end{proof}
\begin{rem}
We do not need a smallness condition like \cite[(5)]{IDNLS} in Proposition~\ref{prop:ODE}. 
\end{rem}
Next we explain a concrete representation formula of the solution
based on the inverse scattering transform. Let us introduce the associated
Ablowitz-Ladik scattering problem
\begin{equation}
X_{n+1}=\mathcal{M}_{n}X_{n},\;\mathcal{M}_{n}=\begin{bmatrix}z & -\bar{R}_{n}\\
R_{n} & z^{-1}
\end{bmatrix},\label{eq:Ablowitz-Ladik}
\end{equation}
where the bar denotes the complex-conjugate\footnote{We quote many formulas from \cite{APT}, in which the complex conjugate
is denoted by $*$. On the other hand, throughout the present paper,
the complex conjugate is denoted by a bar. The $*$'s in $\phi_{n}^{*}(z,t),a^{*}(z)$
etc. are used only for the purpose of distinguishing them from $\phi_{n}(z,t),a(z)$
etc. }. The $t$-part is 
\begin{equation}
\frac{\,d}{\,dt}X_{n}=\begin{bmatrix}-iR_{n-1}\bar{R}_{n}-\frac{i}{2}(z-z^{-1})^{2} & i(z\bar{R}_{n}-z^{-1}\bar{R}_{n-1})\\
i(z^{-1}R_{n}-zR_{n-1}) & iR_{n}\bar{R}_{n-1}+\frac{i}{2}(z-z^{-1})^{2}
\end{bmatrix}X_{n}\label{eq:t-part}
\end{equation}
 and \eqref{eq:IDNLS} is equivalent to the compatibility condition
\ensuremath{\frac{\,d}{\,dt}X_{n+1}=(\frac{\,d}{\,dt}X_{m})_{m=n+1}.}

The condition \eqref{eq:l1} is preserved for $t<\infty$. We can
construct eigenfunctions satisfying \eqref{eq:Ablowitz-Ladik} for
any fixed $t$ (\cite[pp.49-56]{APT}). More specifically, one can
define the eigenfunctions (depending on $t$) \ensuremath{\phi_{n}(z,t),\psi_{n}(z,t)\in\mathcal{O}(|z|>1)\cap\mathcal{C}^{0}(|z|\ge1)}
 and \ensuremath{\psi_{n}^{*}(z,t),\phi_{n}^{*}(z,t)\in\mathcal{O}(|z|<1)\cap\mathcal{C}^{0}(|z|\le1)}
 such that 
\begin{align*}
 & \phi_{n}(z,t)\sim z^{n}\begin{bmatrix}1\\
0
\end{bmatrix},\hspace{-5em} & \phi_{n}^{*}(z,t)\sim z^{-n}\begin{bmatrix}0\\
1
\end{bmatrix} & \quad\mbox{ as }n\to-\infty,\\
 & \psi_{n}(z,t)\sim z^{-n}\begin{bmatrix}0\\
1
\end{bmatrix},\hspace{-5em} & \psi_{n}^{*}(z,t)\sim z^{n}\begin{bmatrix}1\\
0
\end{bmatrix} & \quad\mbox{ as }n\to\infty.
\end{align*}
On the circle $C\colon|z|=1$, there exist unique functions $a(z),\,a^{*}(z)$,
$b(z)=b(z,t)$, $b^{*}(z)=b^{*}(z,t)$ for which 
\begin{align*}
 & \phi_{n}(z,t)=b(z,t)\psi_{n}(z,t)+a(z)\psi_{n}^{*}(z,t),\\
 & \phi_{n}^{*}(z,t)=a^{*}(z)\psi_{n}(z,t)+b^{*}(z,t)\psi_{n}^{*}(z,t)
\end{align*}
holds. It is known that that $a(z)$ and $a^{*}(z)$ are independent
of $t$. They can be represented as Wronskians of the eigenfunctions
and it can be shown that 
\begin{align*}
a(z)\in\mathcal{O}(|z|>1)\cap\mathcal{C}^{0}(|z|\ge1),\hspace{2em} & a^{*}(z)\in\mathcal{O}(|z|<1)\cap\mathcal{C}^{0}(|z|\le1),\\
a^{*}(z)=\bar{a}(1/\bar{z})\;(0<|z|\le1),\hspace{2em} & b^{*}(z)=-\bar{b}(1/\bar{z})\;(|z|=1).
\end{align*}
Moreover, we have $a(z)\to1(z\to\infty)$ and $a^{*}(z)\to1(z\to0)$.

\textit{We assume that $a(z)$ and $a^{*}(z)$ never vanish on the
unit circle. } Their zeros in $|z|>1$ and $|z|<1$ are called \textit{eigenvalues}.
The numbers and the locations of eigenvalues are time-independent.
\textit{We assume that the eigenvalues are all simple}. If $a(z_{j})=0$
and $a^{*}(z_{\ell}^{*})=0$, then we have 
\[
\phi_{n}(z_{j})=b_{j}\psi_{n}(z_{j}),\;\phi_{n}^{*}(z_{\ell}^{*})=b_{\ell}^{*}\psi_{n}^{*}(z_{\ell}^{*})
\]
for some complex constants $b_{j}$ and $b_{\ell}^{*}$. We set 
\[
C_{j}=C_{j}(t)=\frac{b_{j}}{\frac{d}{dz}a(z_{j})},\;C_{\ell}^{*}=C_{\ell}^{*}(t)=\frac{b_{\ell}^{*}}{\frac{d}{dz}a^{*}(z_{\ell}^{*})}
\]
and refer to them as \textit{the norming constants} associated with
the eigenvalues $z_{j}$ and $z_{\ell}^{*}$ respectively.

The following proposition can be found in \cite[p.67]{APT}.
\begin{prop}
\label{prop:symmetry}The eigenvalues come in quartets 
\[
\{\pm z_{j},\pm\bar{z}_{j}^{-1}\}_{j=1}^{J},
\]
where $|z_{j}|>1$. The norming constant associated with $-z_{j}$
(resp. $-z_{j}^{*}=-\bar{z}_{j}^{-1}$) is equal to that associated
with $+z_{j}$ (resp. $+z_{j}^{*}=\bar{z}_{j}^{-1}$). Moreover we
have 
\[
C_{j}^{*}=\bar{z}_{j}^{-2}\bar{C_{j}},
\]
where $C_{j}$ (resp. $C_{j}^{*}$) is the the norming constant associated
with $\pm z_{j}$ (resp. $\pm z_{j}^{*}=\pm\bar{z}_{j}^{-1}$).

Set $\omega_{j}=(z_{j}-z_{j}^{-1})^{2}/2$, $\bar{\omega}_{j}=(\bar{z}_{j}-\bar{z}_{j}^{-1})^{2}/2$.
Then the time evolution of the norming constants is given by  
\begin{equation}
C_{j}(t)=C_{j}(0)\exp(2i\omega_{j}t),\,C_{j}^{*}(t)=C_{j}^{*}(0)\exp(-2i\bar{\omega}_{j}t).\label{eq:normingconstants}
\end{equation}
\end{prop}
We have the characterization equation 
\[
|a(z,t)|^{2}+|b(z,t)|^{2}=c_{-\infty}(\ge1)
\]
on $|z|=1$. We can define the \textit{reflection coefficient} 
\begin{equation}
r(z,t)=\frac{b(z,t)}{a(z,t)},\;|z|=1.
\end{equation}
It has the property $r(-z,t)=-r(z,t)$.

Assume $\{R_{n}(0)\}$ is rapidly decreasing in the sense that \eqref{eq:l1}
holds ($\{R_{n}(0)\}\in\ell^{1,p}$) for any $p\in\mathbb{N}.$ Then
$\{R_{n}(t)\}$ is also rapidly decreasing for any $t$. Due to the
construction in \cite[pp.49-56]{APT}, the eigenfunctions $\phi_{n},\phi_{n}^{*},\psi_{n}$
and $\psi_{n}^{*}$ are smooth on $C\colon|z|=1$. Hence $a,b$ and
$r=r(z,t)$ are also smooth there.

The time evolution of $r(z,t)$ according to \eqref{eq:t-part} is
given by 
\begin{equation}
r(z,t)=r(z)\exp\left(it(z-z^{-1})^{2}\right)=r(z)\exp\left(it(z-\bar{z})^{2}\right),
\end{equation}
where $r(z)=r(z,0)$. Notice that $(z-\bar{z})^{2}$ is real if $|z|=1$. 

Set $c_{n}=\prod_{k=n}^{\infty}(1+|R_{k}|^{2})$. Following \cite[(3.2.94)]{APT},
we set  
\[
m(z)=m(z;n,t)=\begin{cases}
\begin{bmatrix}1 & 0\\
0 & c_{n}
\end{bmatrix}\begin{bmatrix}\dfrac{1}{a(z)}z^{-n}\phi_{n}(z,t), & z^{n}\psi_{n}(z,t)\end{bmatrix} & \text{in}\;|z|>1,\\[10pt]
\begin{bmatrix}1 & 0\\
0 & c_{n}
\end{bmatrix}\begin{bmatrix}z^{-n}\psi_{n}^{*}(z,t), & \dfrac{1}{a^{*}(z)}z^{n}\phi_{n}^{*}(z,t)\end{bmatrix} & \text{in}\;|z|<1.
\end{cases}
\]
It is meromorphic in $|z|\ne1$ with poles $\pm z_{j}$ and $\pm\bar{z}_{j}^{-1}$
and satisfies $m(z)\to I$ as $z\to\infty$. In terms of $m(z)$,
the pole conditions \cite[(3.2.93)]{APT} are, in view of \cite[(3.2.87)]{APT},
\begin{align}
\mathrm{Res}(m(z);\pm z_{j}) & =\lim_{z\to\pm z_{j}}m(z)\begin{bmatrix}0 & 0\\
z_{j}^{-2n}C_{j}(t) & 0
\end{bmatrix},\label{eq:polecondition1}\\
\mathrm{Res}(m(z);\pm\bar{z}_{j}^{-1}) & =\lim_{z\to\pm\bar{z}_{j}^{-1}}m(z)\begin{bmatrix}0 & \bar{z}_{j}^{-2n-2}\bar{C}_{j}(t)\\
0 & 0
\end{bmatrix}\label{eq:polecondition2}
\end{align}
for $j=1,2,\dots,J$. The jump condition is given by 
\begin{align}
 & m_{+}(z)=m_{-}(z)v(z)\;\mbox{\,\ on\,}\;C\colon|z|=1,\label{eq:originalRHP1}\\
 & v(z)=v(z,t)=\begin{bmatrix}1+|r(z,t)|^{2}\; & \;z^{2n}\bar{r}(z,t)\\
z^{-2n}r(z,t)\; & \;1
\end{bmatrix}\nonumber \\
 & \hspace{1.7em}=e^{-(it/2)(z-z^{-1})^{2}\mathrm{ad\,}\sigma_{3}}\begin{bmatrix}1+|r(z)|^{2}\; & \;z^{2n}\bar{r}(z)\\
z^{-2n}r(z)\; & \;1
\end{bmatrix},\label{eq:originalRHP2}\\
 & m(z)\to I\;\mbox{\,\ as\,}\;z\to\infty.\label{eq:originalRHP3}
\end{align}
Here $m_{+}$ and $m_{-}$ are the boundary values from the \textit{outside}
and \textit{inside} of $C$ respectively ($C$ is oriented clockwise
following the convention in \cite{APT}.) We employ the usual notation
$\sigma_{3}=\mathrm{diag\,}(1,-1)$, $a^{\mathrm{ad\,}\sigma_{3}}Q=a^{\sigma_{3}}Qa^{-\sigma_{3}}$.
\begin{rem}
The jump matrix $v(z)$ in \eqref{eq:originalRHP2} is different from
that of \cite{IDNLS} in that $\bar{r}(z)$ is replaced with $-\bar{r}(z).$
Hence $|r(z)|^{2}=\bar{r}(z)r(z)$ is replaced with $-|r(z)|^{2}$.
Other quantities should be modified accordingly.
\end{rem}
The solution $\{R_{n}\}=\{R_{n}(t)\}$ to \eqref{eq:IDNLS} can be
obtained from the $(2,1)$-component of $m(z)$ by the reconstruction
formula (\cite[(3.2.91c)]{APT}) 
\begin{equation}
R_{n}(t)=-\left.\frac{\,d}{\,dz}m(z)_{21}\right|_{z=0}.\label{eq:Rnreconstruction}
\end{equation}

The following proposition can be found in \cite[p.83]{APT}. 
\begin{prop}
\label{prop:m_0}Assume $r(z)\equiv0$ (the potential is reflectionless),
$J=1$ (hence $j=1$) and let $z_{1}=\exp(\alpha_{1}+i\beta_{1}),\alpha_{1}>0$,
be one of the quartet of eigenvalues. Then the RHP \eqref{eq:polecondition1}-\eqref{eq:originalRHP3}
has a unique solution. We denote it by $m_{0}(z)$. The solution $R_{n}(t)$
to \eqref{eq:IDNLS} obtained from $m_{0}(z)$ through \eqref{eq:Rnreconstruction}
is the bright 1-soliton solution $R_{n}(t)=\mathrm{BS}(n,t;z_{1},C_{1}(0))$,
where 
\begin{eqnarray}
\mathrm{BS}(n,t;z_{1},C_{1}(0)) & = & \frac{C_{1}(0)}{|C_{1}(0)|}\exp\bigl(-i[2\beta_{1}(n+1)-2wt]\bigr)\nonumber \\
 &  & \times\sinh(2\alpha_{1})\mathrm{sech}[2\alpha_{1}(n+1)-2vt-\theta].\label{eq:one-soliton}
\end{eqnarray}
Here 
\begin{align*}
 & v=-\sinh(2\alpha_{1})\sin(2\beta_{1}),\quad w=\cosh(2\alpha_{1})\cos(2\beta_{1})-1,\\
 & \theta=\log|C_{1}(0)|-\log\sinh(2\alpha_{1}).
\end{align*}
\end{prop}
\begin{proof}
The unique solvability is proved by using the argument of \cite[pp.72-76]{APT}
and \cite[(3.2.102), (3.2.103))]{APT}. The expression \eqref{eq:one-soliton}
is nothing but \cite[(3.3.143b)]{APT}. 
\end{proof}
Let us introduce 
\[
\varphi=\varphi(z)=\varphi(z;n,t)=\frac{1}{2}it(z-z^{-1})^{2}-n\log z,
\]
so that the jump matrix $v(z)$ in \eqref{eq:originalRHP2} is given
by 
\begin{equation}
v=v(z)=e^{-\varphi\,\mathrm{ad}\sigma_{3}}\begin{bmatrix}1+|r(z)|^{2} & \quad\bar{r}(z)\\
r(z) & \quad1
\end{bmatrix}.\label{eq:v}
\end{equation}
Moreover, we have $\varphi(z_{j})=i\omega_{j}t-n\log z_{j}$ and 
\begin{align}
 & z_{j}^{-2n}C_{j}(t)=C_{j}(0)\exp[2\varphi(z_{j})],\label{eq:exp2varphi}\\
 & \mathrm{Re\,}\varphi(z_{j})=\alpha_{j}t\bigl[\mathrm{tw}(z_{j})-n/t\bigr],\label{eq:tw}\\
 & \mathrm{tw}(z_{j})=-\alpha_{j}^{-1}\sinh(2\alpha_{j})\sin(2\beta_{j}),\label{eq:def_of_tw}
\end{align}
where
\[
z_{j}=\exp(\alpha_{j}+i\beta_{j}),\;\alpha_{j}>0.
\]
Notice the equivalence
\[
\mathrm{Re\,}\varphi(z_{j})>0\Leftrightarrow\mathrm{tw}(z_{j})>n/t.
\]

\begin{rem}
The bright soliton $\mathrm{BS}$ in \eqref{eq:one-soliton} is a
traveling wave with a $\mathrm{sech}$ profile with velocity $\mathrm{tw}(z_{1})$
modulated by a complex carrier wave. Notice that solitons corresponding
to different eigenvalues can have the same velocity. We need a generic
condition in order to avoid anomalies caused by this fact. Namely
\textsl{we assume that $\mathrm{tw}(z_{j})$'s  are mutually distinct.}
It is equivalent to saying that there is at most only one $j$ such
that $\mathrm{Re\,}\varphi(z_{j})=0$ when $n/t$ is fixed. \vspace{0.5em}
\end{rem}
\noindent\textbf{Assumptions (A)} \quad{}\textit{We have made the
following three generic assumptions: }
\begin{itemize}
\item \textit{$a(z)$ never vanishes on the unit circle.  It implies that
$a^{*}(z)$ never vanishes there either. }
\item \textit{The eigenvalues are all simple. }
\item \textsl{$\mathrm{tw}(z_{j})$'s }\textit{are mutually distinct. We
may assume that $\mathrm{tw}(z_{j})<\mathrm{tw}(z_{j+1})$ for any
$j$ without loss of generality. }
\end{itemize}
\textit{They are assumed throughout the present paper. See the appendix
for counter-examples showing that they are not trivial. The first
and the second are assumed in \cite{APT}.}

\vspace{0.5em}
Soliton collision and phase shift in the reflectionless case are studied
in \cite{APT} by a formal calculation. We will give a rigorous argument
based on the Riemann-Hilbert technique. It encompasses the case of
non-zero reflection. 
\begin{lem}
\label{lem:modulus1}

If $|a|=1$, we have 
\[
\mathrm{BS}(n,t;z_{j},aC_{j}(0))=a\mathrm{BS}(n,t;z_{j},C(0)).
\]
The replacement of $C_{j}(0)$ by $aC_{j}(0)$ does not change the
value of $\theta$ in \eqref{eq:one-soliton}. It causes phase shift
in the carrier wave $|C_{1}(0)|^{-1}C_{1}(0)\exp\bigl(-i[2\beta_{1}(n+1)-2wt]\bigr)$
only. In other words, the right-hand side remains a 1-soliton.
\end{lem}

\section{Reduction}

Let $d\negmedspace>\negmedspace0$ be sufficiently small so that the
intervals $[\mathrm{tw}(z_{j})-d,\mathrm{tw}(z_{j})+d],1\le j\le J$,
are mutually disjoint. In other words, the minimum of $|\mathrm{tw}(z_{j})-\mathrm{tw}(z_{k})|\,(j\ne k)$
exceeds $2d$. For each $(n,t)$, there is at most one index $j$
such that $-d\leq\mathrm{tw}(z_{j})-n/t\leq d$. 

For any complex number $a$ and any positive number $\varepsilon$,
let $C(a,\varepsilon)$ and $D(a,\varepsilon)$ be the circle $|z-a|=\varepsilon$
(oriented counterclockwise) and the open disk $|z-a|<\varepsilon$
respectively. 
\begin{prop}
\textup{{[}removal of poles{]} }\label{prop:blowup} Suppose that
$m(z)$ is the solution to the RHP \eqref{eq:polecondition1}-\eqref{eq:originalRHP2}.
For any subset $\sigma$ of $\left\{ 1,2,\dots,J\right\} $, let $\hat{m}(z)$
be defined by 
\[
\hat{m}(z)=\begin{cases}
m(z)\begin{bmatrix}1 & 0\\
-\dfrac{z_{j}^{-2n}C_{j}(t)}{z\mp z_{j}} & 1
\end{bmatrix} & \text{in}\;D(\pm z_{j},\varepsilon),\\[20pt]
m(z)\begin{bmatrix}1 & -\dfrac{\bar{z}_{j}^{-2n-2}\bar{C}_{j}(t)}{z\mp\bar{z}_{j}^{-1}}\\
0 & 1
\end{bmatrix} & \text{in}\;D(\pm\bar{z}{}_{j}^{-1},\varepsilon)
\end{cases}
\]
for each $j\in\sigma$. Here $\varepsilon$ is a sufficiently small
positive constant. Set $\hat{m}(z)=m(z)$ elsewhere. Then $\hat{m}(z)$
is holomorphic near $z=\pm z_{j},\pm\bar{z}_{j}^{-1}$ for $j\in\sigma$.
Instead, it has jumps along the small circles $C(\pm z_{j},\varepsilon)$
and $C(\pm\bar{z}{}_{j}^{-1},\varepsilon).$ Indeed, $\hat{m}(z)$
is the unique solution to 
\begin{align}
\hat{m}_{+}(z) & =\hat{m}_{-}(z)\begin{bmatrix}1 & 0\\
-\dfrac{z_{j}^{-2n}C_{j}(t)}{z\mp z_{j}} & 1
\end{bmatrix} & \text{on}\;C(\pm z_{j},\varepsilon),\label{eq:polecondition1'}\\
\hat{m}_{+}(z) & =\hat{m}_{-}(z)\begin{bmatrix}1 & -\dfrac{\bar{z}_{j}^{-2n-2}\bar{C}_{j}(t)}{z\mp\bar{z}_{j}^{-1}}\\
0 & 1
\end{bmatrix} & \text{on}\;C(\pm\bar{z}{}_{j}^{-1},\varepsilon)\label{eq:polecondition2'}
\end{align}
for $j\in\sigma$ and \eqref{eq:polecondition1}-\eqref{eq:polecondition2}
for $j\not\in\sigma$ with \eqref{eq:originalRHP1}-\eqref{eq:originalRHP3}. 
\end{prop}
\begin{proof}
Let RHP($\sigma$) be the new problem. It is easy to see that RHP($\sigma$)
is equivalent to the original problem RHP($\emptyset$) for any $\sigma$.
The uniqueness for RHP($\left\{ 1,2\dots,J\right\} $) follows from
\cite[Theorem 7.18]{Deiftbook}. The point is that we are dealing
with $2\times2$ jump matrices whose determinants are equal to 1. 
\end{proof}

\begin{lem}
\label{lem:bilinear}Set $\mathrm{conj}(z)=\bar{z}$ for any complex
number $z$. Then we have
\[
\frac{1}{z_{0}}\frac{\bar{p}^{-1}-z_{0}}{\bar{p}^{-1}-\bar{z_{0}}^{-1}}=\mathrm{conj}\!\left(z_{0}\frac{p-\bar{z_{0}}^{-1}}{p-z_{0}}\right)
\]
for any $p,z_{0}\in\mathbb{C}\setminus\left\{ 0\right\} $. In other
words, for $f(p)=z_{0}(p-\bar{z}_{0}^{-1})/(p-z_{0})$, $f(\bar{p}^{-1})$
is the reciprocal of the complex conjugate of $f(p)$. Moreover we
have $f(p)=(p-\alpha)/(\bar{\alpha}p-1)$ for $\alpha=\bar{z}_{0}^{-1}$.
When $|z_{0}|>1$, $f(p)$ is a bilinear transformation that maps
the disk $|p|<1$ onto itself: $|f(p)|=1$ if $|p|=1$. 
\end{lem}
\begin{prop}
\label{prop:gyakusu} Let $\Gamma$ be an oriented contour and $V(z)$
be a given $2\times2$ matrix on it. Assume $z_{0}\ne0$ and $\pm z_{0},\bar{\pm z}_{0}^{-1}\not\in\Gamma$.
For a sufficiently small constant $\varepsilon>0$, let $\Sigma(z_{0})$
be the union of $C(z_{0},\varepsilon)$, $C(-z_{0},\varepsilon)$,
$C(\bar{z}_{0}^{-1},\varepsilon)$ and $C(-\bar{z}_{0}^{-1},\varepsilon)$.
Consider the following Riemann-Hilbert problem on $\Gamma\cup\Sigma(z_{0})$:
\begin{align*}
M_{+}(z) & =M_{-}(z)V(z) & \text{on}\;\Gamma,\\
M_{+}(z) & =M_{-}(z)\begin{bmatrix}1 & 0\\
-\dfrac{A}{z\mp z_{0}} & 1
\end{bmatrix} & \text{on}\;C(\pm z_{0},\varepsilon),\\
M_{+}(z) & =M_{-}(z)\begin{bmatrix}1 & -\dfrac{\bar{z}_{0}^{-2}\bar{A}}{z\mp\bar{z}_{0}^{-1}}\\
0 & 1
\end{bmatrix} & \text{on}\;C(\pm\bar{z}_{0}^{-1},\varepsilon),\\
M(z) & \to I & \text{as}\;z\to\infty.
\end{align*}
Set 
\[
R(z,z_{0})=\frac{z_{0}^{2}(z^{2}-\bar{z}_{0}^{-2})}{z^{2}-z_{0}^{2}}.
\]
Then the RHP above is equivalent to the following one:
\begin{align*}
\tilde{M}_{+}(z) & =\tilde{M}_{-}(z)D(z)^{-1}V(z)D(z) & \text{on}\;\Gamma,\\
D(z) & =\begin{bmatrix}R(z,z_{0})^{-1} & 0\\
0 & R(z,z_{0})
\end{bmatrix} & \text{on}\;\Gamma,\\
\tilde{M}_{+}(z) & =\tilde{M}_{-}(z)\begin{bmatrix}1 & -R(z,z_{0})^{2}\dfrac{z\mp z_{0}}{A}\\
0 & 1
\end{bmatrix} & \text{on}\;C(\pm z_{0},\varepsilon),\\
\tilde{M}_{+}(z) & =\tilde{M}_{-}(z)\begin{bmatrix}1 & 0\\
-R(z,z_{0})^{-2}\dfrac{z\mp\bar{z}_{0}^{-1}}{\bar{z}_{0}^{-2}\bar{A}} & 1
\end{bmatrix} & \text{on}\;C(\pm\bar{z}_{0}^{-1},\varepsilon),\\
\tilde{M}(z) & \to I & \text{as}\;z\to\infty.
\end{align*}
 One can add pole conditions. If the original problem has pole conditions
\begin{align}
\mathrm{Res}(M(z);\pm p) & =\lim_{z\to\pm p}M(z)\begin{bmatrix}0 & 0\\
p^{-2n}C & 0
\end{bmatrix},\label{eq:polecondition1-1}\\
\mathrm{Res}(M(z);\pm\bar{p}^{-1}) & =\lim_{z\to\pm\bar{p}^{-1}}M(z)\begin{bmatrix}0 & \bar{p}^{-2n-2}\bar{C}\\
0 & 0
\end{bmatrix},\label{eq:polecondition2-1}
\end{align}
where $\pm p$ and $\pm\bar{p}^{-1}$ do not belong to the closure
of $D(z_{0},\varepsilon)\cup D(\bar{z}_{0}^{-1},\varepsilon),$ then
the revised conditions are
\begin{align}
\mathrm{Res}(\tilde{M}(z);\pm p) & =\lim_{z\to\pm p}\tilde{M}(z)\begin{bmatrix}0 & 0\\
p^{-2n}\tau C & 0
\end{bmatrix},\label{eq:polecondition1-1-1}\\
\tau & =R(\pm p,z_{0})^{-2}=\left(\frac{p^{2}-z_{0}^{2}}{z_{0}^{2}(p^{2}-\bar{z_{0}}^{-2})}\right)^{2},\\
\mathrm{Res}(\tilde{M}(z);\pm\bar{p}^{-1}) & =\lim_{z\to\pm\bar{p}^{-1}}\tilde{M}(z)\begin{bmatrix}0 & \bar{p}^{-2n-2}\bar{\tau}\bar{C}\\
0 & 0
\end{bmatrix}.\label{eq:polecondition2-1-1}
\end{align}
In other words, $\tau C$ plays the role of the norming constant in
the new problem. 
\end{prop}
\begin{proof}
Set $\tilde{M}(z)=\mathrm{diag}(z_{0}^{2},z_{0}^{-2})M(z)D(z)$, where
\[
D(z)=\begin{cases}
\begin{bmatrix}1 & -\dfrac{z\mp z_{0}}{A}\\
\dfrac{A}{z\mp z_{0}} & 0
\end{bmatrix}\begin{bmatrix}R(z,z_{0})^{-1} & 0\\
0 & \,R(z,z_{0})
\end{bmatrix} & \text{in}\;D(\pm z_{0},\varepsilon),\\[20pt]
\begin{bmatrix}0 & \dfrac{\bar{z}_{0}^{-2}\bar{A}}{z\mp\bar{z}_{0}^{-1}}\\
-\dfrac{z\mp\bar{z}_{0}^{-1}}{\bar{z}_{0}^{-2}\bar{A}} & 1
\end{bmatrix}\begin{bmatrix}R(z,z_{0})^{-1} & 0\\
0 & R(z,z_{0})
\end{bmatrix} & \text{in}\;D(\pm\bar{z}_{0}^{-1},\varepsilon),\\[20pt]
\begin{bmatrix}R(z,z_{0})^{-1} & 0\\
0 & R(z,z_{0})
\end{bmatrix} & \text{elsewhere}.
\end{cases}
\]
Notice that $\pm z_{0}$ and $\pm\bar{z}_{0}^{-1}$ are removable
singularities and that $D(z)$ is holomorphic except on $\Sigma(z_{0})$. 

By Lemma \ref{lem:bilinear}, $R(\pm\bar{p}^{-1},z_{0})$ is the reciprocal
of the complex conjugate of $R(\pm p,z_{0})$. In the derivation of
\eqref{eq:polecondition2-1-1} we use the fact that the complex conjugate
$\bar{\tau}$ of $\tau$ has the expression 
\[
\bar{\tau}=R(\pm\bar{p}^{-1},z_{0})^{2}.
\]
Notice that $\mathrm{diag}(z_{0}^{2},z_{0}^{-2})$ is not on the right
but on the left of $M(z)$ in the definition of $\tilde{M}(z)$. It
has no effect on the jump conditions and the pole conditions. It is
there in order to ensure that $\tilde{M}(z)\to I$ as $z\to\infty$. 
\end{proof}
If $|A|$ is very large in Proposition \ref{prop:gyakusu} above,
then the jump matrices on $\Sigma(z_{0})$ in the latter RHP are very
close to the identity matrix.

We introduce
\begin{align*}
S= & \left\{ k;\,\mathrm{tw}(z_{k})>n/t+d\right\} ,\\
T(z)= & T(z,n/t)=\prod_{k\in S}R(z,z_{k})=\prod_{k\in S}\frac{z_{k}^{2}(z^{2}-\bar{z}_{k}^{-2})}{z^{2}-z_{k}^{2}},\quad T(\infty)=\prod_{k\in S}z_{k}^{2},\\
D_{0}(z)= & \mathrm{diag}\left[T(z)^{-1},\,T(z)\right].
\end{align*}
We set $T(z)=1$ if $S$ is empty. By Lemma \ref{lem:bilinear}, $T(\bar{p}^{-1})$
is the reciprocal of the complex conjugate of $T(p)$. In particular,
we have $|T(z)|=1$ on $|z|=1$. 

\section{The region $|n|<2t$}

We study the asymptotic behavior of $R_{n}(t)$ as $t\to\infty$ in
the region defined by 
\begin{equation}
|n|\le(2-V_{0})t,\quad V_{0}\;\mbox{is a constant with\;}0<V_{0}<2.\label{eq:region}
\end{equation}
We have introduced $V_{0}$ in order to ensure that the uniformity
of the estimates. Other regions will be studied later in Section \ref{sec:Other-regions}. 

We follow closely \cite{IDNLS} and \cite{IDNLS2} in which we studied
the defocusing case. If $|n|<2t$, the function $\varphi(z)=2^{-1}it(z-z^{-1})^{2}-n\log z$
has four saddle points $z=S_{k}\,(k=1,2,3,4)$ on $|z|=1$, where
\begin{align}
 & S_{1}=e^{-\pi i/4}A,\,S_{2}=e^{-\pi i/4}\bar{A},\,S_{3}=-S_{1},\,S_{4}=-S_{2},\\
 & A=2^{-1}\bigl(\sqrt{2+n/t\,}-i\sqrt{2-n/t\,}\,\bigr),
\end{align}
and we set $S_{k\pm4}=S_{k}$ by convention. Let $\delta(z)=\delta_{n/t}(z)=\delta(z;n,t)$,
analytic in $|z|\ne1$, be the solution to the Riemann-Hilbert problem
\begin{align}
 & \delta_{+}(z)=\delta_{-}(z)(1+|r(z)|^{2})\;\mbox{ on}\,\arc{S_{1}S_{2}}\;\mbox{and}\;\arc{S_{3}S_{4}},\label{eq:delta1-1}\\
 & \delta_{+}(z)=\delta_{-}(z)\;\mbox{ on}\,\arc{S_{2}S_{3}}\;\mbox{and}\;\arc{S_{4}S_{1}},\label{eq:delta2-1}\\
 & \delta(z)\to1\;\mbox{ as}\;z\to\infty,\label{eq:delta3-1}
\end{align}
where $\arc{S_{j}S_{k}}$ is the minor arc $\subset\left\{ |z|=1\right\} $
joining $S_{j}$ and $S_{k}$ and the \textit{outside} of $\left\{ |z|=1\right\} $
is the plus side. 

This problem can be uniquely solved by the formula 
\begin{equation}
\delta(z)=\exp\left(\frac{-1}{2\pi i}\left[\int_{S_{1}}^{S_{2}}+\int_{S_{3}}^{S_{4}}\right](\tau-z)^{-1}\log(1+|r(\tau)|^{2})\,d\tau\right),\label{eq:deltaintegral}
\end{equation}
where the contours are the arcs $\subset\left\{ |z|=1\right\} $.
We have $\delta(-z)=\delta(z)$ and $\delta'(0)=0$ because $r(-\tau)=-r(\tau)$.
Notice that $0<\delta(0)\le1$. We have $\delta(0)=1$ if and only
if $r(z)$ vanishes identically on the arcs. 

Under Assumptions (A), we have:
\begin{thm}
\label{thm:main1}Let $V_{0}$ be a constant with $0<V_{0}<2$. Assume
that the initial value satisfies the rapid decrease condition $\{R_{n}(0)\}\in\bigcap_{p=0}^{\infty}\ell^{1,p}$
(i.e. \eqref{eq:l1} holds for any $p\in\mathbb{N}$). Then in the
region $|n|\le(2-V_{0})t$, the asymptotic behavior of the solution
to \eqref{eq:IDNLS} is as follows:\textbf{}\\
\textbf{(soliton case)} In the region $-d\le\mathrm{tw}(z_{s})-n/t\le d,\,s\in\left\{ 1,\dots,J\right\} ,$
where $d$ is sufficiently small, we have 

\begin{align*}
R_{n}(t) & =\mathrm{BS}\left(n,t;z_{s},\delta(0)\delta(z_{s})^{-2}p_{s}T(z_{s})^{-2}C_{s}(0)\right)+O(t^{-1/2}),\\
p_{s} & =\prod_{k>s}\,z_{k}^{2}\bar{z}_{k}^{-2},\;\hspace{2em}T(z_{s})=\prod_{k>s}\frac{z_{k}^{2}(z_{s}^{2}-\bar{z}_{k}^{-2})}{z_{s}^{2}-z_{k}^{2}}.
\end{align*}
We have $S=\left\{ k;\,k>s\right\} $, hence the expression of $T(z_{s})$
above. \\
\textbf{(solitonless case)} If $\left\{ \mathrm{tw}(z_{j});j=1,\dots,J\right\} \cap[n/t-d,n/t+d]=\emptyset,$
then there exist $C_{j}=C_{j}(n/t)\in\mathbb{C}$ and $p_{j}=p_{j}(n/t),q_{j}=q_{j}(n/t)\in\mathbb{R}$
($j=1,2$) depending only on the ratio $n/t$ such that 
\begin{equation}
R_{n}(t)=\sum_{j=1}^{2}C_{j}t^{-1/2}e^{-i(p_{j}t+q_{j}\log t)}+O(t^{-1}\log t)\hspace{0.5em}\mbox{as}\hspace{0.5em}t\to\infty.\label{eq:main}
\end{equation}
The behavior of each term in the sum is decaying oscillation of order
$t^{-1/2}$ as $t\to\infty$ while $n/t$ is fixed. The symbol $O$
represents an asymptotic estimate which is uniform with respect to
$(t,n)$ satisfying $|n|\le(2-V_{0})t$. 
\end{thm}
\begin{proof}
The soliton case is shown in Proposition \ref{prop: soliton}. The
solitonless case can be proved in almost the same way as \cite{IDNLS}.
See Remark \ref{rem:proofsolitonless}. 
\end{proof}
\begin{rem}
\label{rem:delta}We see that $\delta(z)=\delta(z;n,t)$ is determined
by $z$ and $n/t$. When we are interested in a particular ray $n/t=\mathrm{const.}$,
we suppress the dependence on $n/t$. On the other hand, when we are
interested in multiple rays, we prefer the notation $\delta_{n/t}(z)$.
\end{rem}
We set $d=\frac{1}{2}\log\delta(0)$ and introduce the following two
matrices: 
\begin{eqnarray*}
 & \Delta(z)=\begin{bmatrix}\delta(z) & 0\\
0 & \delta(z)^{-1}
\end{bmatrix}, & \,\\
 & \Delta(0)^{1/2}=\begin{bmatrix}\delta(0)^{1/2} & 0\\
0 & \delta(0)^{-1/2}
\end{bmatrix} & =e^{d\sigma_{3}}.
\end{eqnarray*}

Set $\tilde{\delta}(z)=\bar{\delta}(\bar{z}^{-1})=\overline{\delta(\bar{z}^{-1})}$.
Then it is the unique solution to the problem below:
\begin{align*}
 & \tilde{\delta}_{+}(z)=\tilde{\delta}_{-}(z)/(1+|r(z)|^{2})\;\mbox{ on}\,\arc{S_{1}S_{2}}\;\mbox{and}\;\arc{S_{3}S_{4}},\\
 & \tilde{\delta}_{+}(z)=\tilde{\delta}_{-}(z)\;\mbox{ on}\,\arc{S_{2}S_{3}}\;\mbox{and}\;\arc{S_{4}S_{1}},\\
 & \tilde{\delta}(0)=1.
\end{align*}
The solution formula is
\[
\tilde{\delta}(z)=\delta(0)\exp\left(\frac{1}{2\pi i}\left[\int_{S_{1}}^{S_{2}}+\int_{S_{3}}^{S_{4}}\right](\tau-z)^{-1}\log(1+|r(\tau)|^{2})\,d\tau\right)=\delta(0)\delta(z)^{-1}.
\]
So we get $\overline{\delta(\bar{z}^{-1})}=\delta(0)\delta(z)^{-1}$.
Since $\delta(0)>0$, we have
\begin{equation}
\delta(\bar{z}^{-1})=\delta(0)\bar{\delta}(z)^{-1}.\label{eq:deltareflection}
\end{equation}

With Propositions \ref{prop:blowup} and \ref{prop:gyakusu} in mind,
we define a matrix $D_{1}(z)$ as follows. For each $j$ with $\mathrm{tw}(z_{j})>n/t+d$,
we define 
\begin{eqnarray*}
 & D_{1}(z)= & \begin{cases}
\begin{bmatrix}1 & -\dfrac{z\mp z_{j}}{z_{j}^{-2n}C_{j}(t)}\\
\dfrac{z_{j}^{-2n}C_{j}(t)}{z\mp z_{j}} & 0
\end{bmatrix}D_{0}(z)\Delta(0)^{1/2} & \text{in}\;D(\pm z_{j},\varepsilon),\\[20pt]
\begin{bmatrix}0 & \dfrac{\bar{z}_{j}^{-2n-2}\bar{C}_{j}(t)}{z\mp\bar{z}_{j}^{-1}}\\
-\dfrac{z\mp\bar{z}_{j}^{-1}}{\bar{z}_{j}^{-2n-2}\bar{C}_{j}(t)} & 1
\end{bmatrix}D_{0}(z)\Delta(0)^{1/2} & \text{in}\;D(\pm\bar{z}_{j}^{-1},\varepsilon)
\end{cases}
\end{eqnarray*}
and set $D_{1}(z)=D_{0}(z)\Delta(0)^{1/2}$ elsewhere. Notice that
we have 
\[
|z_{j}^{-2n}C_{j}(t)|=|C_{j}(0)|\exp[2\mathrm{Re\,}\varphi(z_{j})]=|C_{j}(0)|\exp[2\alpha_{j}t\left\{ \mathrm{tw}(z_{j})-n/t\right\} ]
\]
by \eqref{eq:exp2varphi} and \eqref{eq:tw}.
\begin{prop}
\label{prop:onlyonepole} Let $\sigma$ in Proposition~\ref{prop:blowup}
be defined by $\sigma=\left\{ 1,2,\dots,J\right\} \setminus\left\{ s\right\} .$
Here $s$ is such that $-d\leq\mathrm{tw}(z_{s})-n/t\leq d$.\footnote{If there is no such $s$, set $\sigma=\left\{ 1,2,\dots,J\right\} .$}
Then $\left\{ \pm z_{s},\pm\bar{z}_{s}^{-1}\right\} $ is the only
quartet of poles of $\hat{m}(z)$. Set $\tilde{m}(z)=\mathrm{diag}(T(\infty),T(\infty)^{-1})\Delta(0)^{-1/2}\hat{m}(z)D_{1}(z)$.
Then \\
\textup{(i)} For each $j$ with $\mathrm{tw}(z_{j})-n/t<-d$, we have
\begin{align*}
\tilde{m}_{+}(z) & =\tilde{m}_{-}(z)I_{\mathrm{exp}}^{-}(z;\pm z_{j})\hspace{1em}on\;C(\pm z_{j},\varepsilon),\\
\text{where}\; & I_{\mathrm{exp}}^{-}(z;\pm z_{j})=\begin{bmatrix}1 & 0\\
-\dfrac{z_{j}^{-2n}\delta(0)T(z)^{-2}C_{j}(t)}{z\mp z_{j}} & 1
\end{bmatrix},\\
\tilde{m}_{+}(z) & =\tilde{m}_{-}(z)I_{\mathrm{exp}}^{-}(z;\pm\bar{z}_{j}^{-1})\text{\hspace{1em}on}\;C(\pm\bar{z}_{j}^{-1},\varepsilon),\\
\text{where}\; & I_{\mathrm{exp}}^{-}(z;\pm\bar{z}_{j}^{-1})=\begin{bmatrix}1 & -\dfrac{\bar{z}_{j}^{-2n-2}\delta(0)^{-1}T(z)^{2}\bar{C}_{j}(t)}{z\mp\bar{z}_{j}^{-1}}\\
0 & 1
\end{bmatrix}.
\end{align*}
\textup{(ii)} For each $j$ with $\mathrm{tw}(z_{j})-n/t>d$, we have
\begin{align*}
\tilde{m}_{+}(z) & =\tilde{m}_{-}(z)I_{\mathrm{exp}}^{+}(z;\pm z_{j})\hspace{1em}\text{on}\;C(\pm z_{j},\varepsilon),\\
\text{where}\; & I_{\mathrm{exp}}^{+}(z;\pm z_{j})=\begin{bmatrix}1 & -\dfrac{\delta(0)^{-1}(z\mp z_{j})}{z_{j}^{-2n}T(z)^{-2}C_{j}(t)}\\
0 & 1
\end{bmatrix},\\
\tilde{m}_{+}(z) & =\tilde{m}_{-}(z)I_{\mathrm{exp}}^{+}(z;\pm\bar{z}_{j}^{-1})\hspace{1em}\text{on}\;C(\pm\bar{z}_{j}^{-1},\varepsilon),\\
\text{where}\; & I_{\mathrm{exp}}^{+}(z;\pm\bar{z}_{j}^{-1})=\begin{bmatrix}1 & 0\\
-\dfrac{\delta(0)(z\mp\bar{z}_{j}^{-1})}{\bar{z}_{j}^{-2n-2}T(z)^{2}\bar{C}_{j}(t)} & 1
\end{bmatrix}.
\end{align*}
\textup{(iii)} If $j=s$, the pole conditions become 
\begin{align}
\mathrm{Res}(\tilde{m}(z);\pm z_{s}) & =\lim_{z\to\pm z_{s}}\tilde{m}(z)I_{\mathrm{res}}(z_{s}),\nonumber \\
\text{where}\; & I_{\mathrm{res}}(z_{s})=\begin{bmatrix}0 & 0\\
z_{s}^{-2n}\delta(0)T(z_{s})^{-2}C_{s}(t) & 0
\end{bmatrix},\label{eq:Ires1}\\
\mathrm{Res}(\tilde{m}(z);\pm\bar{z}_{s}^{-1}) & =\lim_{z\to\pm\bar{z}_{s}^{-1}}\tilde{m}(z)I_{\mathrm{res}}(\bar{z}_{s}^{-1}),\nonumber \\
where\; & I_{\mathrm{res}}(\bar{z}_{s}^{-1})=\begin{bmatrix}0 & \bar{z}_{s}^{-2n-2}\delta(0)^{-1}\bar{T}(z_{s})^{-2}\bar{C}_{s}(t)\\
0 & 0
\end{bmatrix}.\label{eq:Ires2}
\end{align}
Notice that any $j\in\left\{ 1,2,\dots,J\right\} $ satisfies one
of \textup{(i)}, \textup{(ii)} or \textup{(iii)}. It is possible that
no $j$ satisfies \textup{(iii)}. \\
\textup{(iv)} On $C\colon|z|=1$ (clockwise), we have 
\begin{equation}
\tilde{m}_{+}(z)=\tilde{m}_{-}(z)\Delta(0)^{-1/2}D_{0}(z)^{-1}v(z)D_{0}(z)\Delta(0)^{1/2}.\label{eq:D_0^-1vD_0}
\end{equation}
\textup{(v)} $\tilde{m}(z)\to I$ as $z\to\infty$. 
\end{prop}
\begin{proof}
Apply Proposition~\ref{prop:gyakusu} repeatedly when $z_{0}=\pm z_{j}$
for $j\in\left\{ \mathrm{tw}(z_{j})-n/t>d\right\} $. We have (v)
owing to the factor $\mathrm{diag}(T(\infty),T(\infty)^{-1})\Delta(0)^{-1/2}$.
It has no effect on the jump and the pole conditions. We have used
the fact that $T(\bar{z}_{j}^{-1})^{-1}=\bar{T}(z_{j})$. 
\end{proof}
Because of \eqref{eq:exp2varphi} and \eqref{eq:tw}, $C_{j}(t)$
and $\bar{C}_{j}(t)$ are exponentially decreasing (resp. increasing)
as $t\to\infty$ if $\mathrm{tw}(z_{j})-n/t<0$ (resp. $>0$). The
jump matrices $I_{\mathrm{exp}}^{\pm}(z;\pm z_{j})$ and $I_{\mathrm{exp}}^{\pm}(z;\pm\bar{z}_{j}^{-1})$
in (i) and (ii) of Proposition \ref{prop:onlyonepole} are exponentially
close to $I$. The case (iii) is about a soliton. 

Compare $I_{\mathrm{res}}(z_{s})$ and $I_{\mathrm{res}}(\bar{z}_{s}^{-1})$.
The symmetry in the pair \eqref{eq:polecondition1}-\eqref{eq:polecondition2},
which is essential in Proposition \ref{prop:m_0}, is lost in the
sense that $\delta(0)^{-1}$ is not the complex conjugate of $\delta(0).$
Symmetry will be recovered in \eqref{eq:20150723a}-\eqref{eq:20150723b}
after the $\Delta(z)$-conjugation. The fact is that we have introduced
$\delta(0)$ and $\delta(0)^{-1}$ as a precaution in order to perform
the $\Delta(z)$-conjugation without breaking symmetry. 

Conjugating our Riemann-Hilbert problem in Proposition \ref{prop:onlyonepole}
by $\Delta(z)$ leads to the following factorization problem for $\tilde{m}\Delta^{-1}$,
in which $\pm[\mathrm{tw}(z_{j})-n/t]>d,\sigma=\pm1$: 
\begin{align}
 & (\tilde{m}\Delta^{-1})_{+}(z)=(\tilde{m}\Delta^{-1})_{-}(z)(\Delta_{-}\Delta(0)^{-1/2}D_{0}^{-1}vD_{0}\Delta(0)^{1/2}\Delta_{+}^{-1})(z)\hspace{1em}\mathrm{on}\,C,\label{eq:DeltaRHP1}\\
 & (\tilde{m}\Delta^{-1})_{+}(z)=(\tilde{m}\Delta^{-1})_{-}(z)\left(\Delta I_{\mathrm{exp}}^{\pm}(z;\sigma z_{j})\Delta^{-1}\right)(z)\hspace{1em}\mathrm{on}\,C(\sigma z_{j},\varepsilon),\label{eq:exponentially1}\\
 & (\tilde{m}\Delta^{-1})_{+}(z)=(\tilde{m}\Delta^{-1})_{-}(z)\left(\Delta I_{\mathrm{exp}}^{\pm}(z;\sigma\bar{z}_{j}^{-1})\Delta^{-1}\right)(z)\hspace{1em}\mathrm{on}\,C(\sigma\bar{z}_{j}^{-1},\varepsilon),\label{eq:exponentially2}\\
 & \mathrm{Res}(\tilde{m}(z)\Delta(z)^{-1};\pm z_{s})=\lim_{z\to\pm z_{s}}\tilde{m}(z)\Delta(z)^{-1}\Delta(z_{s})I_{\mathrm{\mathrm{res}}}(z_{s})\Delta(z_{s})^{-1},\label{eq:DeltaRHP2}\\
 & \mathrm{Res}(\tilde{m}(z)\Delta(z)^{-1};\pm\bar{z}_{s}^{-1})=\lim_{z\to\pm\bar{z}_{s}^{-1}}\tilde{m}(z)\Delta(z)^{-1}\Delta(\bar{z}_{s}^{-1})I_{\mathrm{\mathrm{res}}}(\bar{z}_{s}^{-1})\Delta(\bar{z}_{s}^{-1})^{-1},\label{eq:DeltaRHP3}\\
 & \tilde{m}\Delta^{-1}\to I\quad(z\to\infty).\label{eq:DeltaRHP4}
\end{align}
Notice that the jump matrices in \eqref{eq:exponentially1} and \eqref{eq:exponentially2}
are exponentially close to $I$ as $t$ tends to infinity. We calculate
the jump matrix in \eqref{eq:DeltaRHP1}. On $C\colon|z|=1$(clockwise)
we have $|T(z)|=1,T(z)^{-1}=\overline{T(z)}$ and \eqref{eq:v} implies
\begin{align}
 & \Delta(0)^{-1/2}D_{0}(z)^{-1}v(z)D_{0}(z)\Delta(0)^{1/2}\nonumber \\
 & =e^{-\varphi\,\mathrm{ad\sigma_{3}}}\begin{bmatrix}1+|r(z)|^{2} & \delta(0)^{-1}\bar{r}(z)T(z)^{2}\\
\delta(0)r(z)\bar{T}(z)^{2} & 1
\end{bmatrix},\quad z\in C.\label{eq:D_0^-1vD_0-1}
\end{align}
Since $d=\frac{1}{2}\log\delta(0)$, we obtain 
\begin{align*}
 & \Delta_{-}\Delta(0)^{-1/2}D_{0}^{-1}vD_{0}\Delta(0)^{1/2}\Delta_{+}^{-1}\\
 & =\Delta(0)^{-1/2}\left\{ \Delta_{-}D_{0}^{-1}vD_{0}\Delta_{+}^{-1}\right\} \Delta(0)^{1/2}\\
 & =e^{-(\varphi+d)\mathrm{ad}\sigma_{3}}\begin{bmatrix}\delta_{+}^{-1}\delta_{-}(1+|r|^{2}) & \delta_{+}\delta_{-}\bar{r}T^{2}\\
\delta_{+}^{-1}\delta_{-}^{-1}r\bar{T}^{2} & \delta_{+}\delta_{-}^{-1}
\end{bmatrix},\quad z\in\,C.
\end{align*}
On the other hand, \eqref{eq:Ires1} implies 
\begin{equation}
\Delta(z_{s})I_{\mathrm{res}}(z_{s})\Delta(z_{s})^{-1}=\begin{bmatrix}0 & 0\\
z_{s}^{-2n}\delta(0)\delta(z_{s})^{-2}T(z_{s})^{-2}C_{s}(t) & 0
\end{bmatrix}\label{eq:20150723a}
\end{equation}
and we get by \eqref{eq:deltareflection} and \eqref{eq:Ires2}
\begin{eqnarray}
\Delta(\bar{z}_{s}^{-1})I_{\mathrm{res}}(\bar{z}_{s}^{-1})\Delta(\bar{z}_{s}^{-1})^{-1} & =\begin{bmatrix}0 & \bar{z}_{s}^{-2n-2}\delta(0)\bar{\delta}(z_{s})^{-2}\bar{T}(z_{s})^{-2}\bar{C}_{s}(t)\\
0 & 0
\end{bmatrix}.\label{eq:20150723b}
\end{eqnarray}
Therefore\textsl{ $\delta(0)\delta(z_{s})^{-2}T(z_{s})^{-2}C_{s}(t)$
plays the role of the norming constant in \eqref{eq:polecondition1}-\eqref{eq:polecondition2}. }

Now, we rewrite \eqref{eq:DeltaRHP1} by choosing the counterclockwise
orientation (the \textit{inside} being the plus side) on $\arc{S_{2}S_{3}}$
and $\arc{S_{4}S_{1}}$ and the clockwise orientation on $\arc{S_{1}S_{2}}$
and $\arc{S_{3}S_{4}}$. The circle $|z|=1$ with this new orientation
is denoted by $\widetilde{C}$ and \eqref{eq:DeltaRHP1} is replaced
with 
\begin{align}
 & (\tilde{m}\Delta^{-1})_{+}(z)=(\tilde{m}\Delta^{-1})_{-}(z)\tilde{v}(z),\quad z\in\widetilde{C}\label{eq:mdelta^-1_1}
\end{align}
for another 2$\times$2 matrix $\tilde{v}$. Notice that \eqref{eq:exponentially1}-\eqref{eq:DeltaRHP4}
remain unchanged. We have 
\begin{align*}
\tilde{v} & =\tilde{v}(z)=e^{-(\varphi+d)\mathrm{ad}\sigma_{3}}\left(\begin{bmatrix}1 & \;0\\
\delta_{-}^{-2}r\bar{T}^{2}/(1+|r|^{2})\; & \;1
\end{bmatrix}\begin{bmatrix}1\; & \delta_{+}^{2}\bar{r}T^{2}/(1+|r|^{2})\\
0\; & \;1
\end{bmatrix}\right) & \mbox{ }
\end{align*}
on $\arc{S_{1}S_{2}}\cup\arc{S_{3}S_{4}}$ and 
\[
\tilde{v}=e^{-(\varphi+d)\mathrm{ad}\sigma_{3}}\left(\begin{bmatrix}1\; & \;0\\
-\delta^{-2}r\bar{T}^{2}\; & \;1
\end{bmatrix}\begin{bmatrix}1\; & -\delta^{2}\bar{r}T^{2}\\
0\; & \;1
\end{bmatrix}\right)
\]
 on $\arc{S_{2}S_{3}}\cup\arc{S_{4}S_{1}}$. 

Set 
\begin{align}
\rho & =-\bar{r}T^{2}/(1+|r|^{2}) & \mbox{\quad on\;}\arc{S_{1}S_{2}}\cup\arc{S_{3}S_{4}},\label{eq:rho1}\\
 & =\bar{r}T^{2} & \mbox{\quad on\;}\arc{S_{2}S_{3}}\cup\arc{S_{4}S_{1}}.\label{eq:rho2}
\end{align}
Then $\tilde{v}$ admits the unified expression 
\begin{equation}
\tilde{v}=e^{-(\varphi+d)\mathrm{ad}\sigma_{3}}\left(\begin{bmatrix}1\; & \;0\\
-\delta_{-}^{-2}\bar{\rho}\; & \;1
\end{bmatrix}\begin{bmatrix}1\; & \;-\delta_{+}^{2}\rho\\
0\; & \;1
\end{bmatrix}\right)\label{eq:unifiedexpression}
\end{equation}
on any of the arcs. 

\begin{rem}
\label{rem: whatisdifferent}What is different from \cite[p.773]{IDNLS}
is that $\varphi$, $r$, $\bar{r}$ and $1-|r|^{2}$ are replaced
with $\varphi+d$, $r\bar{T}^{2},-\bar{r}T^{2}$ and $1+|r|^{2}$.
Recall that $|T|=1$ on $\tilde{C}.$ The additional term $+d$ and
the action of $\exp(d\mathrm{\,ad}\sigma_{3})$ can be treated by
using the technique of \cite[(18)]{IDNLS2}. 
\end{rem}

\section{A Riemann-Hilbert problem on a new contour}

\label{sec:newcontour} In this section, we introduce a new contour
and formulate a new Riemann-Hilbert problem, which is equivalent to
the problem \eqref{eq:mdelta^-1_1}, \eqref{eq:exponentially1}-\eqref{eq:DeltaRHP4}.
The new jump matrix admits a certain lower/upper factorization which
will be the basis of the integral representation given later. 

\begin{figure}
\includegraphics[scale=0.3]{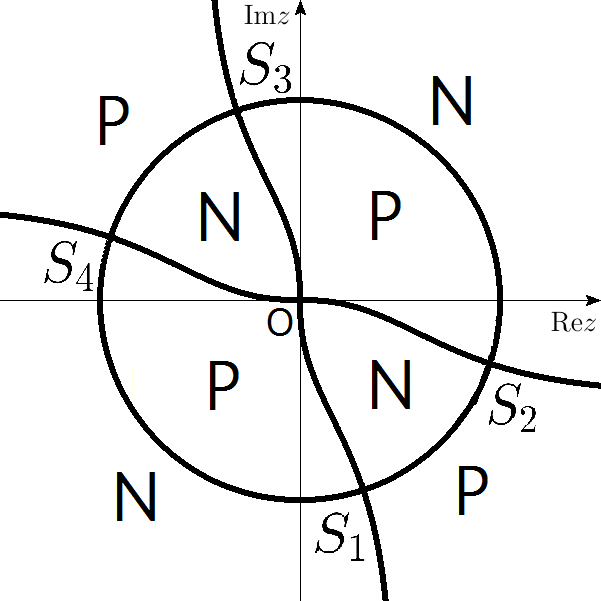}

\caption{Signs of $\mathrm{Re}\,\varphi$}
\label{fig:phase}

\end{figure}

The signs of $\mathrm{Re}\,\varphi$ are shown in Figure \ref{fig:phase}:
P and N stand for 'positive' and 'negative' respectively and $S_{j}$'s
are the saddle points.  Let $\Sigma$ be the contour (including the
dotted and dashed parts) in Figure 2. The + signs indicate the plus
side. The black squares are the poles $\pm z_{s},\pm\bar{z}_{s}^{-1}$
in \eqref{eq:DeltaRHP2} and \eqref{eq:DeltaRHP3}. The small circles
are centered at $\pm z_{j}$ and $\pm\bar{z}_{j}^{-1}$ for some $j\ne s$.
We can bend $\Sigma$ so that the black squares and the small circles
are in $\Omega_{1}\cup\Omega_{2}.$ There may be more quartets of
small circles, but they are omitted in the figure. The orientations
of the small circles are irrelevant because the jump matrices along
them are exponentially close to the identity matrix. 

The large circle is $\tilde{C}$. The union of the quartet(s) of the
small circles is called $C^{\circ}$. The dotted part and the dashed
part are called $L$ and $L'$ respectively. We have $\Sigma=\tilde{C}\cup C^{\circ}\cup L\cup L'$.
Notice that $\mathrm{Re\,\varphi>0}$ on $L\setminus\left\{ \textrm{the saddle points}\right\} $
and that $\mathrm{Re\,\varphi<0}$ on $L'\setminus\left\{ \textrm{the saddle points}\right\} $. 

On each arc joining adjacent saddle points, we have the decomposition
\[
\rho=R+h_{I}+h_{II},\;\bar{\rho}=\bar{R}+\bar{h}_{I}+\bar{h}_{II}.
\]
This is a 'curved version' of the decomposition in \cite{DZ} and
its construction is a variant of that given in \cite{IDNLS,IDNLS2}.
Here we just state what is necessary to understand the present paper.
The leading parts are $R$ and $\bar{R}$ and the limit of $R(z)$
and $\bar{R}(z)$ as $z$ tends to a saddle point along an arc coincides
with that of $\rho(z)$ and $\bar{\rho}(z)$ respectively. The other
parts, $h_{I},h_{II},\bar{h}_{I}$ and $\bar{h}_{II}$ are small in
the following sense. First, $|e^{-2\varphi}h_{I}|$ and $|e^{2\varphi}\bar{h}_{I}|$
are estimated by any negative power of $t$. Second, $h_{II}$ and
$R$ (resp. $\bar{h}_{II}$ and $\bar{R}$ ) can be analytically continued
to $\left\{ \text{\text{Re\,}}\varphi>0\right\} $ (resp. $\left\{ \text{\text{Re\,}}\varphi<0\right\} $)
and $|e^{-2\varphi}h_{II}|$ (resp. $|e^{2\varphi}\bar{h}_{II}|$
) is estimated by any negative power of $t$ on $L$ (resp. $L'$).
Lastly, $|e^{-2\varphi}R|$ (resp. $|e^{2\varphi}\bar{R|}$) decay
exponentially on $L$ (resp. on $L'$) except in small neighborhoods
of the saddle points. 

\begin{figure}
\includegraphics[width=12cm]{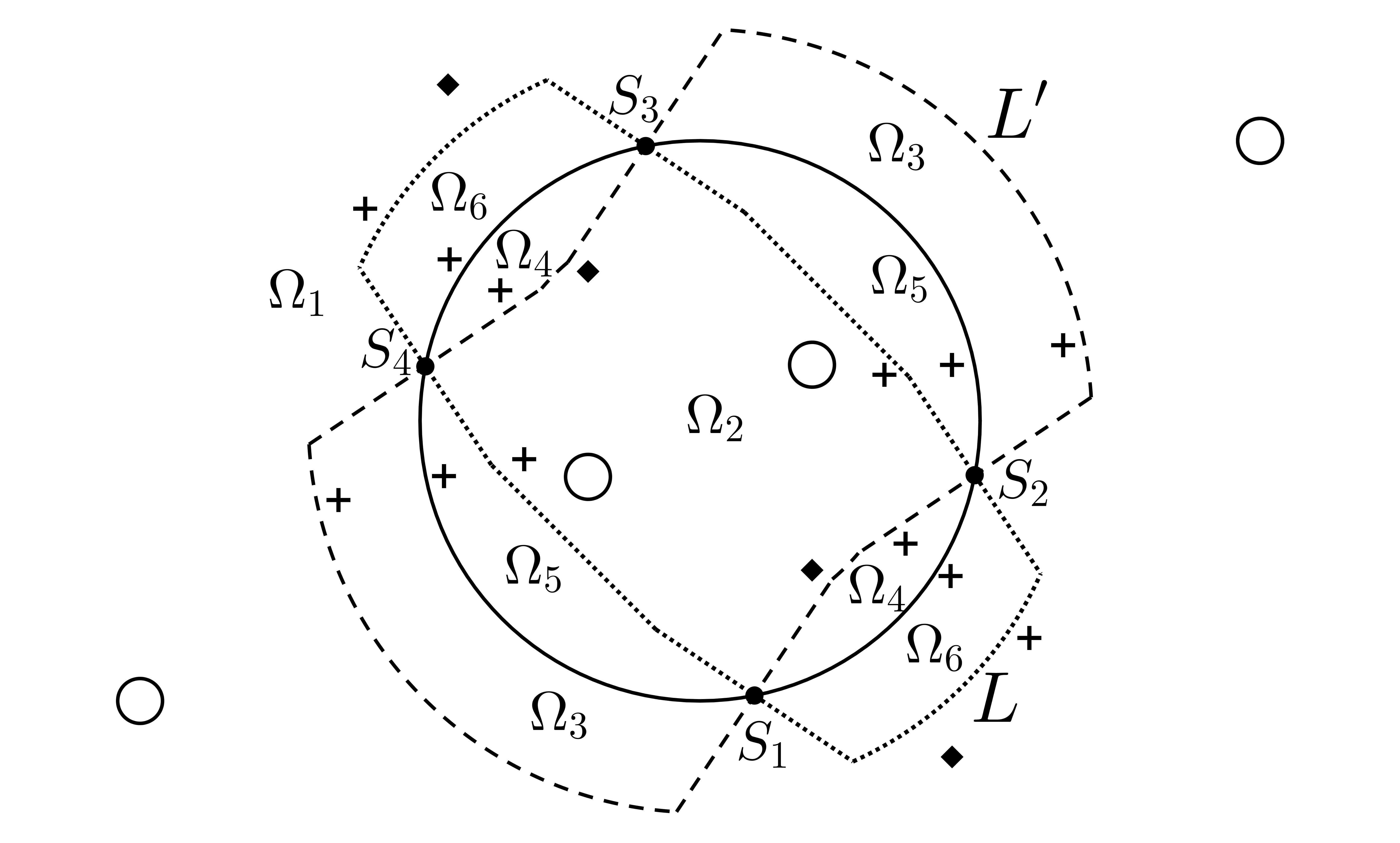}\label{fig:Sigma}\caption{the contour $\Sigma$}
\end{figure}

We introduce the following matrices: 
\begin{align*}
 & b_{+}^{0}=\delta_{+}^{\mathrm{ad}\sigma_{3}}e^{-(\varphi+d)\mathrm{ad}\sigma_{3}}\begin{bmatrix}1 & \quad-h_{I}\\
0 & \quad1
\end{bmatrix},\; & b_{+}^{a}=\delta_{+}^{\mathrm{ad}\sigma_{3}}e^{-(\varphi+d)\mathrm{ad}\sigma_{3}}\begin{bmatrix}1\; & \;-h_{II}-R\\
0\; & \;1
\end{bmatrix},\\
 & b_{-}^{0}=\delta_{-}^{\mathrm{ad}\sigma_{3}}e^{-(\varphi+d)\mathrm{ad}\sigma_{3}}\begin{bmatrix}1 & \quad0\\
\bar{h}_{I} & \quad1
\end{bmatrix},\; & b_{-}^{a}=\delta_{-}^{\mathrm{ad}\sigma_{3}}e^{-(\varphi+d)\mathrm{ad}\sigma_{3}}\begin{bmatrix}1\; & \;0\\
\bar{h}_{II}+\bar{R}\; & \;1
\end{bmatrix}.
\end{align*}
Notice that $b_{\pm}^{a}$ can be analytically continued to $\left\{ \pm\mathrm{Re\,}\varphi>0\right\} $.
By \eqref{eq:unifiedexpression}, we have 
\begin{equation}
\tilde{v}=(b_{-}^{a})^{-1}(b_{-}^{\circ})^{-1}b_{+}^{\circ}b_{+}^{a}\label{eq:unifiedexpression2}
\end{equation}
on any of the arcs. Set \ensuremath{b_{-}^{\sharp}=I,b_{-}^{0},b_{-}^{a}}
 and \ensuremath{b_{+}^{\sharp}=b_{+}^{a},b_{+}^{0},I}
 on $L,\tilde{C},L'$ respectively. On $\Sigma\setminus C^{\circ}$,
we set
\[
v^{\sharp}=v^{\sharp}(z)=(b_{-}^{\sharp})^{-1}b_{+}^{\sharp}.
\]
We have $v^{\sharp}=b_{+}^{a},(b_{-}^{0})^{-1}b_{+}^{0},(b_{-}^{a})^{-1}$
on $L,\tilde{C},L'$ respectively. On the remaining part $C^{\circ}$,
let $v^{\sharp}$ be equal to the jump matrices in \eqref{eq:exponentially1}
and \eqref{eq:exponentially2}. As a replacement for $\tilde{m}$
in Proposition \ref{prop:onlyonepole}, or rather $\tilde{m}\Delta^{-1}$
in \eqref{eq:mdelta^-1_1}, \eqref{eq:exponentially1}-\eqref{eq:DeltaRHP4},
we define a new unknown matrix $m^{\sharp}$ by 
\begin{align}
m^{\sharp} & =\tilde{m}\Delta^{-1}, & z\in\Omega_{1}\cup\Omega_{2},\label{eq:msharpdef1}\\
 & =\tilde{m}\Delta^{-1}(b_{-}^{a})^{-1}, & z\in\Omega_{3}\cup\Omega_{4},\label{eq:msharpdef2}\\
 & =\tilde{m}\Delta^{-1}(b_{+}^{a})^{-1}, & z\in\Omega_{5}\cup\Omega_{6}.\label{eq:msharpdef3}
\end{align}
It is the unique solution to the Riemann-Hilbert problem 
\begin{align}
 & m_{+}^{\sharp}(z)=m_{-}^{\sharp}(z)v^{\sharp}(z),\hspace{1.2em}z\in\Sigma,\label{eq:msharp1}\\
 & \mathrm{Res}(m^{\sharp};\pm z_{s})=\lim_{z\to\pm z_{s}}m^{\sharp}(z)\Delta(z_{s})I_{\mathrm{\mathrm{res}}}(z_{s})\Delta(z_{s})^{-1},\label{eq:msharp2}\\
 & \mathrm{Res}(m^{\sharp};\pm\bar{z}_{s}^{-1})=\lim_{z\to\pm\bar{z}_{s}^{-1}}m^{\sharp}(z)\Delta(\bar{z}_{s}^{-1})I_{\mathrm{\mathrm{res}}}(\bar{z}_{s}^{-1})\Delta(\bar{z}_{s}^{-1})^{-1},\label{eq:msharp3}\\
 & m^{\sharp}(z)\to I\hspace{1.2em}\mbox{as\;}z\to\infty.\label{eq:msharp4}
\end{align}
See \eqref{eq:20150723a}-\eqref{eq:20150723b} for concrete expressions
of matrices in  \eqref{eq:msharp2}-\eqref{eq:msharp3}. We shall
employ \ensuremath{w_{\pm}^{\sharp}=\pm(b_{\pm}^{\sharp}-I)},
 \ensuremath{w^{\sharp}=w_{+}^{\sharp}+w_{-}^{\sharp}}
. We have \ensuremath{v^{\sharp}=(I-w_{-}^{\sharp})^{-1}(I+w_{+}^{\sharp})=(I+w_{-}^{\sharp})(I+w_{+}^{\sharp}).}
 Notice that $v^{\sharp}(z)$ is defined on $C^{\circ}$ in terms
of $I_{\mathrm{exp}}^{\pm}(z;z_{j})$ and $I_{\mathrm{exp}}^{\pm}(z;\bar{z}_{j}^{-1})$.
It is exponentially close to $I$ on $C^{\circ}$ as $t\to\infty$. 

Let us derive a reconstruction formula in terms of $m_{21}^{\sharp}$.
Near $z=0$, we have $m_{21}^{\sharp}(z)=\tilde{m}_{21}(z)\delta(z)^{-1}$,
$\tilde{m}_{21}(z)=\hat{m}_{21}(z)\delta(0)T(z)^{-1}T(\infty)^{-1}$,
$\hat{m}_{21}(z)=m_{21}(z)$. Therefore we obtain 
\[
m_{21}(z)=\tilde{m}{}_{21}(z)\delta(0)^{-1}T(\infty)T(z)=m_{21}^{\sharp}(z)\delta(0)^{-1}\delta(z)T(\infty)T(z).
\]
Set $p_{s}=T(0)T(\infty)$. Then we have 
\begin{equation}
p_{s}=\prod_{k>s}\,z_{k}^{2}\bar{z}_{k}^{-2},\quad|p_{s}|=1.\label{eq:p_s}
\end{equation}
Since $\delta(z)$ and $T(z)$ are even functions, we get 
\begin{equation}
R_{n}(t)=-\left.\frac{\,d}{\,dz}m(z)_{21}\right|_{z=0}=-T(0)T(\infty)\frac{\,dm_{21}^{\sharp}}{\,dz}(0)=-p_{s}\frac{\,dm_{21}^{\sharp}}{\,dz}(0).\label{eq:Rnreconst}
\end{equation}

\section{Modified Cauchy kernel and the Beals-Coifman formula\label{sec:Modified-Cauchy-kernel}}

Set 
\[
g(\zeta,z;a)=\frac{(z-a)(z+a)}{(\zeta-a)(\zeta+a)(\zeta-z)}=\frac{1}{\zeta-z}-\frac{z+a}{2a(\zeta-a)}+\frac{z-a}{2a(\zeta+a)}.
\]
We have $g(\zeta,\pm a;a)=0$. Next set $h_{1}(\zeta,z)=g(\zeta,z;\bar{z}_{s}^{-1}),\,h_{2}=g(\zeta,z;z_{s})$.
We have $h_{1}(\zeta,\pm\bar{z}_{s}^{-1})=h_{2}(\zeta,\pm z_{s})=0.$
We introduce the modified Cauchy kernel 
\[
\Omega(\zeta,z)=\frac{1}{2\pi i}\begin{bmatrix}h_{1}(\zeta,z) & 0\\
0 & h_{2}(\zeta,z)
\end{bmatrix}.
\]
For any $2\times2$ matrix $f=f(z)$, the second columns of $(f\Omega)(\pm z_{s})$
and the first columns of $(f\Omega)(\pm\bar{z}_{s}^{-1})$ are zero
for any $\zeta$. We define the modified Cauchy operator $C^{\Omega}$
by $(C^{\Omega}f)(z)=\int_{\Sigma}f(\zeta)\Omega(\zeta,z)\,d\zeta$.
We have 
\begin{eqnarray}
(C^{\Omega}f)(\zeta,\pm z_{s})=\begin{bmatrix}* & 0\\*
 & 0
\end{bmatrix},\label{eq:COmegapolecondition1}\\
(C^{\Omega}f)(\zeta,\pm\bar{z}_{s}^{-1})=\begin{bmatrix}0 & *\\
0 & *
\end{bmatrix}.\label{eq:COmegapolecondition2}
\end{eqnarray}
 The boundary values of $(C^{\Omega}f)(z)$ on $\Sigma$ are denoted
by 
\[
(C_{\pm}^{\Omega}f)(z)=\int_{\Sigma}f(\zeta)\Omega(\zeta,z_{\pm})\,d\zeta=\lim_{\genfrac{}{}{0pt}{}{z'\to z}{z'\in\{\pm\textrm{\scriptsize-side of }\Sigma\}}}\int_{\Sigma}f(\zeta)\Omega(\zeta,z')\,d\zeta,\;z\in\Sigma.
\]
We have $C_{+}^{\Omega}-C_{-}^{\Omega}=\mathrm{identity}$. We introduce
the modified Beals-Coifman operator \ensuremath{C_{w^{\sharp}}^{\Omega}\colon L^{2}(\Sigma)\to L^{2}(\Sigma)}
 by 
\begin{equation}
C_{w^{\sharp}}^{\Omega}f=C_{+}^{\Omega}(fw_{-}^{\sharp})+C_{-}^{\Omega}(fw_{+}^{\sharp})\label{eq:integopdef}
\end{equation}
for a $2\times2$ matrix-valued function $f$. 

Let $\mu^{\sharp}$ be the solution to the equation 
\begin{equation}
\mu^{\sharp}=m_{0}^{\sharp}+C_{w^{\sharp}}^{\Omega}\mu^{\sharp}.
\end{equation}
Here $m_{0}^{\sharp}$ is obtained from $m_{0}$ by replacing $C_{1}(0)$
with $\delta(0)\delta(z_{s})^{-2}T(z_{s})^{-2}C_{s}(0)$ in Proposition
\ref{prop:m_0}. See \eqref{eq:20150723a} and \eqref{eq:20150723b}.
We have $\mu^{\sharp}=(1-C_{w^{\sharp}}^{\Omega})^{-1}m_{0}^{\sharp}$
(the resolvent exists as is proved in the next section), and 
\begin{equation}
m^{\sharp}(z)=m_{0}^{\sharp}(z)+\int_{\Sigma}\mu^{\sharp}(\zeta)w^{\sharp}(\zeta)\Omega(\zeta,z),\;z\in\mathbb{C}\setminus\left(\Sigma\cup\left\{ \pm z_{s},\pm\bar{z}_{s}^{-1}\right\} \right)\label{eq:msharpdef}
\end{equation}
is the unique solution to the Riemann-Hilbert problem \eqref{eq:msharp1}-\eqref{eq:msharp4}.
Indeed, the pole conditions \eqref{eq:msharp2}-\eqref{eq:msharp3}
follow from \eqref{eq:20150723a}, \eqref{eq:20150723b} and \eqref{eq:COmegapolecondition1}-\eqref{eq:COmegapolecondition2}.
On the other hand, \eqref{eq:msharp1} is satisfied because 
\begin{align*}
m_{+}^{\sharp} & =m_{0}^{\sharp}+C_{+}^{\Omega}(\mu^{\sharp}w^{\sharp})=m_{0}^{\sharp}+C_{+}^{\Omega}(\mu^{\sharp}w_{+}^{\sharp})+C_{+}^{\Omega}(\mu^{\sharp}w_{-}^{\sharp})\\
 & =m_{0}^{\sharp}+\mu^{\sharp}w_{+}^{\sharp}+C_{-}^{\Omega}(\mu^{\sharp}w_{+}^{\sharp})+C_{+}^{\Omega}(\mu^{\sharp}w_{-}^{\sharp})\\
 & =m_{0}^{\sharp}+C_{w^{\sharp}}^{\Omega}(\mu^{\sharp})+\mu^{\sharp}w_{+}^{\sharp}=\mu^{\sharp}+\mu^{\sharp}w_{+}^{\sharp}=\mu^{\sharp}b_{+}^{\sharp}
\end{align*}
and similarly $m_{-}^{\sharp}=\mu^{\sharp}b_{-}^{\sharp}$. By substituting
\eqref{eq:msharpdef} into \eqref{eq:Rnreconst}, we find that
\begin{align}
R_{n}(t) & =p_{s}\biggl[\mathrm{BS}\left(n,t;z_{s},\delta(0)\delta(z_{s})^{-2}T(z_{s})^{-2}C_{s}(0)\right)+E_{n}(t)\biggr],\nonumber \\
E_{n}(t) & =-\int_{\Sigma}z^{-2}\left[\bigl((1-C_{w^{\sharp}}^{\Omega})^{-1}m_{0}^{\sharp}\bigr)(z)w^{\sharp}(z)\right]_{21}\frac{\,dz}{2\pi i}\nonumber \\
 & \quad\,+\int_{\Sigma}\frac{1}{(z-\bar{z}_{s}^{-1})(z+\bar{z}_{s}^{-1})}\left[\bigl((1-C_{w^{\sharp}}^{\Omega})^{-1}m_{0}^{\sharp}\bigr)(z)w^{\sharp}(z)\right]_{21}\frac{\,dz}{2\pi i}.\label{eq:BS+En}
\end{align}

\section{Estimates\label{sec:Estimates}}

In this section we prove the existence of the resolvent $(1-C_{w^{\sharp}}^{\Omega})^{-1}\colon L^{2}(\Sigma)\to L^{2}(\Sigma)$
and give an estimate on the error term $E_{n}(t)$ in \eqref{eq:BS+En}.
Let $C_{\pm}$ be the boundary values of the usual Cauchy integrals.
We introduce the Beals-Coifman operator \ensuremath{C_{w^{\sharp}}\colon L^{2}(\Sigma)\to L^{2}(\Sigma)}
 by 
\[
C_{w^{\sharp}}f=C_{+}(fw_{-}^{\sharp})+C_{-}(fw_{+}^{\sharp})
\]
for a $2\times2$ matrix-valued function $f$. 
\begin{prop}
\label{prop:resolvents}The resolvents $(1-C_{w^{\sharp}})^{-1},(1-C_{w^{\sharp}}^{\Omega})^{-1}\colon L^{2}(\Sigma)\to L^{2}(\Sigma)$
exist for any sufficiently large $t$.
\end{prop}
\begin{proof}
The existence of $(1-C_{w^{\sharp}}^{\Omega})^{-1}$ follows from
that of $(1-C_{w^{\sharp}})^{-1}$, because the difference $C_{w^{\sharp}}^{\Omega}-C_{w^{\sharp}}$
is infinitely small for large $t$. The proof is as follows. We see
that $\left(C_{w^{\sharp}}^{\Omega}-C_{w^{\sharp}}\right)f$ consists
of terms like 
\[
\int_{\Sigma}\frac{\mathrm{const.}fw_{\pm}^{\sharp}}{\zeta\pm z_{s}}d\zeta,\,\int_{\Sigma}\frac{\mathrm{const.}fw_{\pm}^{\sharp}}{\zeta\pm\bar{z}_{s}^{-1}}d\zeta
\]
and that $\pm z_{s},\pm\bar{z}_{s}^{-1}\not\not\in\Sigma$. Since
the $L^{2}(\Sigma)$-norm of $w_{\pm}^{\sharp}$ is $O(t^{-1/4})$
(\cite[\S7.2]{IDNLS}), the $L^{\infty}(\Sigma)$-norm of $\left(C_{w^{\sharp}}^{\Omega}-C_{w^{\sharp}}\right)f$
is $O(t^{-1/4})\times\|f\|_{L^{2}(\Sigma)}$. Since $\Sigma$ is bounded,
the $L^{2}(\Sigma)$-norm of $\left(C_{w^{\sharp}}^{\Omega}-C_{w^{\sharp}}\right)f$
is also $O(t^{-1/4})\times\|f\|_{L^{2}(\Sigma)}$.

Next we show the existence of $(1-C_{w^{\sharp}})^{-1}$. Since $\sup_{C^{\circ}}|w_{\pm}^{\sharp}|=O(t^{-N})$
for any $N$, we have only to prove that $(1-C_{w^{\sharp}})^{-1}\colon L^{2}(\Sigma\setminus C^{\circ})\to L^{2}(\Sigma\setminus C^{\circ})$
exists. Here we abuse the notation $C_{w^{\sharp}}$ to mean an operator
on $L^{2}(\Sigma\setminus C^{\circ})$. Then the necessary argument
is similar to \cite[\S9, \S11]{IDNLS}. We encounter the matrix 
\[
V=e^{-((iz^{2})/4)\mathrm{ad}\,\sigma_{3}}\begin{bmatrix}1 & -\bar{r}(S_{j})\\
-r(S_{j}) & 1+|r(S_{j})|^{2}
\end{bmatrix}=\begin{bmatrix}1 & -\bar{r}(S_{j})e^{-iz^{2}/2}\\
-r(S_{j})e^{iz^{2}/2} & 1+|r(S_{j})|^{2}
\end{bmatrix}
\]
instead of $v^{e,\phi}(z)$ at the bottom of \cite[p.796]{IDNLS}.
We have only to prove the existence of the resolvent of the Beals-Coifman
operator in $L^{2}(\mathbb{R})$ associated with it. Now $|r(S_{j})|$
is not necessarily less than 1. The simple argument based on the Neumann
series as in \cite[(3.94)]{DZ} and \cite[p.797]{IDNLS} is not valid.
We can resort to \cite[Lemma 5.9]{DP} instead. It implies the existence
of $(1-C_{v^{\mathrm{DP}}})^{-1}$ in $L^{2}(\mathbb{R})$, where
\[
v^{\mathrm{DP}}=v^{\mathrm{DP}}(z)=\begin{bmatrix}1+|a|^{2} & \bar{a}e^{i\theta}\\
ae^{-i\theta} & 1
\end{bmatrix},\,\theta=-z^{2}/2,\,a\colon\text{const.},
\]
and $C_{v^{\mathrm{DP}}}$ is the Beals-Coifman operator associated
with any factorization of $v^{\mathrm{DP}}$. Let $\mathbb{R}^{\mathrm{rev}}$
be the contour obtained by reversing the orientation of $\mathbb{R}$.
By \cite[Proposition 2.8]{DZ_Sobolev}, we have 
\[
C_{v^{\mathrm{DP}}}=C_{\left(v^{\mathrm{DP}}\right)^{-1}}\quad\text{in}\,L^{2}(\mathbb{R})=L^{2}(\mathbb{R}^{\mathrm{rev}}).
\]
Notice that $\mathbb{R}^{\mathrm{rev}}$ can be identified with the
conventionally oriented real axis of another copy of of $\mathbb{C}$
via $\mathbb{C}_{z}\to\mathbb{C}_{w},z\mapsto w:=-z$ and that we
have $dz/(z-\xi)=dw/(w-\zeta),\zeta=-\xi$. Since $v^{\mathrm{DP}}(\xi)=v^{\mathrm{DP}}(\zeta)$,
the resolvent $(1-C_{\left(v^{\mathrm{DP}}\right)^{-1}})^{-1}$ exists
in $L^{2}(\mathbb{R})$. Notice that $\left(v^{\mathrm{DP}}\right)^{-1}$
has the same form as $V$. We have proved the existence of the resolvent
of the Beals-Coifman operator associated with $V$.
\end{proof}
Next we can show $E_{n}(t)=O(t^{-1/2})$. The following proposition
is a part of Theorem \ref{thm:main1}.
\begin{prop}
\label{prop: soliton}In the region $\mathrm{(-2<)tw}(z_{s})-d\leq n/t\leq\mathrm{tw}(z_{s})+d(<2)$,
where $d$ is sufficiently small, the solution $R_{n}(t)$ differs
from a soliton only by $O(t^{-1/2})$: 
\[
R_{n}(t)=\mathrm{BS}\left(n,t;z_{s},\delta(0)\delta(z_{s})^{-2}p_{s}T(z_{s})^{-2}C_{s}(0)\right)+O(t^{-1/2}).
\]
\end{prop}
\begin{proof}
Lemma \ref{lem:modulus1} implies that $p_{s}$ times a soliton is
still a soliton. We have 
\begin{eqnarray*}
p_{s}\mathrm{BS}\left(n,t;z_{s},\delta(0)\delta(z_{s})^{-2}T(z_{s})^{-2}C_{s}(0)\right)\\
\mathrm{=BS}\left(n,t;z_{s},\delta(0)\delta(z_{s})^{-2}p_{s}T(z_{s})^{-2}C_{s}(0)\right).
\end{eqnarray*}
By using a change of variables (scaling) as in \cite[pp.798-799]{IDNLS},
we can show that $E_{n}(t)$ in \eqref{eq:BS+En} satisfies $E_{n}(t)=O(t^{-1/2})$.
The calculations about the parabolic cylinder functions only need
minor changes. One important step relies on the fact that the determinant
of the jump matrix is equal to $1$ (\cite[pp.349-350]{DZ}), which
remains true for the focusing discrete NLS.  
\end{proof}
\begin{rem}
\label{rem:proofsolitonless}The proof of the solitonless case of
Theorem \ref{thm:main2} is almost the same as that of the defocusing
case (\cite{IDNLS}). The calculations about the parabolic cylinder
functions only need minor changes. One important step relies on the
fact that the determinant of the jump matrix is equal to $1$ (\cite[pp.349-350]{DZ}),
which remains true.  

One thing to be noted is that now $r(z)$ is accompanied by the factor
$\bar{T}(z)^{2}.$ It has no effect when we consider quantities involving
$|r(z)|^{2}$ because we have $|T(z)|=1$ on $C$. Although $C_{j}$'s
are affected by $T(z)$, $q_{j}$'s are not. See \cite[Theorem 3.1]{IDNLS}. 
\end{rem}

\section{\label{sec:Other-regions}Other regions}

In the preceding sections we considered the region $|n|<2t$. In this
section, we consider two other regions following \cite{IDNLS2}. The
equation \eqref{eq:IDNLS} is invariant under the reflection $n\mapsto-n$.
We may assume $n>0$ without loss of generality.

By using the argument of \cite{IDNLS2}, we can show the following
theorem.
\begin{thm}
\label{thm:main2}Assume that $\mathrm{tw}(z_{s})=2$ for some eigenvalue
$z_{s}$. Then in the region $2t-Mt^{1/3}<n<2t+M't^{1/3}\;(M>0)$,
we have
\[
R_{n}(t)=\mathrm{BS}\left(n,t;z_{s},p_{s}T(z_{s})^{-2}C_{s}(0)\right)+O(t^{-1/3})\quad\text{as}\;t\to\infty.
\]
In the solitonless case, i.e. if $\mathrm{tw}(z_{j})\ne2$ for any
$j$, then the behavior is as follows: let $t_{0}$ be such that $\pi^{-1}\left[\arg r(e^{-\pi i/4})\overline{T(e^{-\pi i/4})}^{2}-2t_{0}\right]$
is an integer. Set $t'=t-t_{0}$, $p'=d+i(-4t'+\pi n)/4$, $\alpha'=[12t'/(6t'-n)]^{1/3}$,
$q'=-2^{-4/3}3^{1/3}(6t'-n)^{-1/3}(2t'-n)$ and $\hat{r}=r(e^{-\pi i/4})\overline{T(e^{-\pi i/4})}^{2}$.
Then we have 
\[
R_{n}(t)=\frac{e^{2p'-\pi i/4}\alpha'}{(3t')^{1/3}}u\Bigl(\frac{4q'}{3^{1/3}};\hat{r},-\hat{r},0\Bigr)+O(t'^{-2/3}).
\]
Here $u(s;\mathrm{p,q,r})$ is a solution of the Painlev\'e II equation
$u''(s)-su(s)-2u^{3}(s)=0$. Its parametrization is given in \cite{DZ}
(and is repeated in \cite{IDNLS2}). 
\end{thm}
\begin{proof}
In order to prove the soliton case, derive a variant of \eqref{eq:msharpdef}
adapted to this case. To estimate the integral, follow the argument
of \cite[\S8]{IDNLS2}. The calculation is somewhat different from
the previous one. This is because we apply a scaling of order $t^{-1/3}$
as in \cite[\S4]{IDNLS2}. Notice that values of $\delta(z)$ are
absent in the third argument of $\mathrm{BS}$. It is because the
signature table of the phase function is good enough from the beginning
and the $\Delta(z)$-conjugation is unnecessary. Indeed, in the proof
of \cite[Lemma 5.9]{DP}, what matters is the fact that the jump matrix
is strictly positive definite (\cite[Lemma 5.2]{DP}). 

The proof of the existence of resolvents must be modified because
the modulus of the reflection coefficient is not necessarily less
than $1$. We can employ a variant of Lemma 5.9 of \cite{DP}. We
replace the phase with $-4i(z^{3}-\textrm{const}.z)$. The change
of the signature table does not spoil the proof. It rather simplifies
the argument since one can employ a simpler factorization without
introducing the function $\delta$ of \cite[(5.12)]{DP}. Moreover,
what matters in the proof of \cite[Lemma 5.9]{DP} is the fact that
the jump matrix is strictly positive definite and phase functions
need not to be quadratic. 

In the solitonless case, we have to modify the argument of \cite{IDNLS2}
slightly. The quantity $d=\frac{1}{2}\log\delta(0)$ can be dealt
with simulaneously with $p'.$ 

Next, notice that $\begin{bmatrix}1-|r|^{2} & -\bar{r}\\
r & 1
\end{bmatrix}$ (in the expression of $v$), $\begin{bmatrix}1 & 0\\
r & 1
\end{bmatrix}$ (in the expression of $J_{23}$) and $\begin{bmatrix}1 & -\bar{r}\\
0 & 1
\end{bmatrix}$ (in the expression of $J_{67}$) in \cite{IDNLS2} must be replaced
with $\begin{bmatrix}1+|r|^{2} & -\bar{r}T^{2}\\
r\bar{T}^{2} & 1
\end{bmatrix}$, $\begin{bmatrix}1 & 0\\
r & 1
\end{bmatrix}$ and $\begin{bmatrix}1 & -\bar{r}T^{2}\\
r\bar{T}^{2} & 1
\end{bmatrix}$ in the present paper. We \textsl{want }to set 
\[
\mathrm{p}=r(e^{-\pi i/4})\overline{T(e^{-\pi i/4})}^{2},\mathrm{q}=-r(e^{-\pi i/4})\overline{T(e^{-\pi i/4})}^{2}=-\bar{r}(e^{-\pi i/4})T(e^{-\pi i/4})^{2},
\]
where $\mathrm{\mathrm{p}}$ and $\mathrm{\mathrm{q}}$ are the parameters
in \cite[(5.33)]{DZ} or \cite[Appendix]{IDNLS2}. Notice that $r(e^{-\pi i/4})\overline{T(e^{-\pi i/4})}^{2}$
must be real. It is possible to reduce our problem to such a case
by a time shift like the one in \cite[\S5]{IDNLS2}. In \cite{IDNLS2},
$r(e^{-\pi i/4})$ was purely imaginary and $T$ was not present. 
\end{proof}
\vspace{1em}
Next we consider the region $|n/t|>2$. 
\begin{thm}
\label{thm:main3}In $2<\mathrm{tw}(z_{s})-d\leq n/t\leq\mathrm{tw}(z_{s})+d$,
where $d$ is sufficiently small, we have 
\[
R_{n}(t)=\mathrm{BS}\left(n,t;z_{s},p_{s}T(z_{s})^{-2}C_{s}(0)\right)+O(n^{-k})\quad\text{as}\;|n|\to\infty
\]
for any positive integer $k$. 

In the solitonless case, i.e. if $\mathrm{tw}(z_{j})\not\in[n/t-d,n/t+d]$
for any $j$, then 
\[
R_{n}(t)=O(n^{-k})\quad\text{as}\;|n|\to\infty
\]
for any positive integer $k$. 
\end{thm}
\begin{proof}
In order to prove the former case, derive a variant of \eqref{eq:msharpdef}
adapted to this case. To estimate the integral, follow the argument
of \cite[\S8]{IDNLS2}. The latter case can be proved in the same
way as \cite[\S8]{IDNLS2}. 
\end{proof}

\section*{Appendix: Counter-examples}

We show that $a(z)$ may vanish on the unit circle $|z|=1$ and that
it may have double zeros. Moreover, we prove that there may be two
eigenvalues corresponding to the same velocity. 

Assume that $R_{n}=0$ for $n\ne0,1,2$. Then we have (\cite[(3.2.24)]{APT})
\begin{equation}
\phi_{n}(z,t)=z^{n}\begin{bmatrix}1\\
0
\end{bmatrix}\,(n\leqq0),\;\psi_{n}(z,t)=z^{-n}\begin{bmatrix}0\\
1
\end{bmatrix}\,(n\geqq3).\label{eq:counter1}
\end{equation}
The $n$-part \eqref{eq:Ablowitz-Ladik} implies 
\begin{equation}
\phi_{3}(z,t)=\mathcal{M}_{2}\mathcal{M}_{1}\mathcal{M}_{0}\begin{bmatrix}1\\
0
\end{bmatrix}=\begin{bmatrix}z^{3}-(R_{0}\bar{R}_{1}+R_{1}\bar{R}_{2})z-R_{0}\bar{R}_{2}z^{-1}\\
\text{omitted}
\end{bmatrix}.\label{eq:counter2}
\end{equation}

We have the Wronskian formula (\cite[(3.2.64c)]{APT}) 
\[
a(z)=c_{n}W(\phi_{n},\psi_{n}),\;c_{n}=\prod_{k=n}^{\infty}(1+|R_{n}|^{2})\geqq1.
\]
Here $n$ is arbitrary. We apply it to the case $n=3$. We employ
\eqref{eq:counter2} and $\psi_{3}$ is calculated by using \eqref{eq:counter1}.
We set $R_{1}=1$. Then we obtain
\[
a(z)=c_{3}z^{-4}f(z^{2}),\;f(x)=x^{2}-(R_{0}+\bar{R}_{2})x-R_{0}\bar{R}_{2}.
\]

It is elementary that $f(x)$ can have any pair of complex numbers,
say $x_{1}$ and $x_{2}$, as zeros if $R_{0}$ and $\bar{R}_{2}$
are suitably chosen. It is enough to choose $R_{0}$ and $\bar{R}_{2}$
so that 
\[
R_{0}+\bar{R}_{2}=x_{1}+x_{2},\;R_{0}\bar{R}_{2}=-x_{1}x_{2}.
\]

The zeros of $a(z)$ are $\pm x_{1}^{1/2},\pm x_{2}^{1/2}$. If we
choose $x_{1}$ and $x_{2}$ properly, the following three phenomena
can occur:
\begin{itemize}
\item $a(z)$ has zeros on $|z|=1$.
\item $a(z)$ has double zeros.
\item $\mathrm{tw}(x_{1}^{1/2})=\mathrm{tw}(x_{2}^{1/2})$.
\end{itemize}
Of course, a general theory of Darboux transformations is preferable.

\end{document}